\documentclass[a4paper,12pt]{article}

\usepackage{amssymb,amsmath,amsthm,mathbbol,mathrsfs}
\usepackage[usenames,dvipsnames]{color}
\usepackage{hyperref}
\usepackage{stmaryrd}
\usepackage{authblk}
\usepackage{framed}
\usepackage{empheq}
\usepackage{slashed}
\usepackage{hyphenat}
\usepackage{apacite}
\usepackage{natbib}
\usepackage{comment}

\usepackage{marginnote}    
\usepackage[left=.8in,right=.8in,top=.8in,bottom=.8in]{geometry}                  

\usepackage{sectsty} 
\allsectionsfont{\sffamily\mdseries\upshape} 
\usepackage{tocloft}

\makeatletter
\renewenvironment{abstract}{%
    \if@twocolumn
      \section*{\abstractname}%
    \else 
      \begin{center}%
        {\bfseries\sffamily\abstractname\vspace{\z@}}
      \end{center}%
      \quotation
    \fi}
    {\if@twocolumn\else\endquotation\fi}
\makeatother

\numberwithin{equation}{section}
5

\hypersetup{
	colorlinks=true,         
	linkcolor=MidnightBlue,          
	citecolor=BrickRed,        
	urlcolor=MidnightBlue            
}

\newcommand{\be}{\begin{equation}}
\newcommand{\ee}{\end{equation}}

\newcommand{\fg}{{\mathfrak{g}}}
\newcommand{\w}{{^{~\text{\tiny P}}\Gamma}}
\newcommand{\fLie}{{\mathbb{L}}}
\newcommand{\fI}{{\mathbb{i}}}
\newcommand{\F}{{\mathcal{Q}}}
\renewcommand{\d}{{\mathrm{d}}}
\newcommand{\D}{{\mathrm{D}}}

\newcommand{\R}{{\mathbb{R}}}
\newcommand{\rr}{{\mathbf{r}}}
\newcommand{\cc}{{\mathbf{c}}}
\newcommand{\SO}{{\mathrm{SO}}}
\newcommand{\HH}{{\mathsf{H}}}
\newcommand{\V}{{\mathsf{V}}}

\newcommand{\G}{{\mathcal{G}}}
\newcommand{\Ad}{{\mathsf{Ad}}}

\newcommand{\Ker}{{\mathsf{Ker}}}

\newcommand{\fF}{{\mathbb{F}}}

\newcommand{\fG}{{\mathrm{Lie}(\G)}}

\renewcommand{\bar}{\overline}
\newcommand{\dd}{{\mathbb{d}}}

\newcommand{\ra}{\rightarrow}
\renewcommand{\hat}{\widehat}

\newcommand{\lbr}{\llbracket}
\newcommand{\rbr}{\rrbracket}

\newtheorem{theo}{Theorem}

\newcommand{\cint}{{\int\kern-.87em{<}}}
\newcommand{\sint}{{\int\kern-.75em{\sim}}}
\newcommand{\fint}{{\int\kern-1.00em{\int}}}

\newcommand{\bb}{\mathbb}

\renewcommand{\#}{\sharp}

\let\oldmarginpar\marginpar
\renewcommand\marginpar[1]{\oldmarginpar{\color{red}\raggedright\footnotesize #1}}

\title{Angular momentum without rotation: \\ Turbocharging relationalism }
\author{
  Henrique Gomes\footnote{University of Cambridge, Trinity College, CB2 1TQ, United Kingdom; \href{mailto:gomes.ha@gmail.com}{gomes.ha@gmail.com}}\, and Sean Gryb\footnote{ University of Groningen, Faculty of Philosophy, Groningen, The Netherlands;  \href{mailto:sean.gry@gmail.com}{sean.gryb@gmail.com} }
 }
 \date{}

\begin{document}
\maketitle 
\abstract{ Newton's rotating bucket pours cold water on the naive relationalist by vividly illustrating how certain rotational effects, particularly those due to non-zero angular momentum, can depend on more than just relations between material bodies. Because of such effects, rotation has played a central role in the absolute-relational debate and poses a particularly difficult challenge to the relationalist.

In this paper, we provide a qualified response to this challenge that significantly weakens  the absolutist position. We present a theory that, contrary to  orthodoxy, can account for all rotational effects without introducing, as the absolutist does, a fixed standard of rotation. Instead, our theory posits a universal $\SO(3)$ charge that plays the role of angular momentum and couples to inter-particle relations via terms commonly seen in standard gauge theories such as electromagnetism and the Standard Model of particle physics.

Our theory makes use of an enriched form of relationalism: it adds an $\SO(3)$ structure to the traditional relational description. Our construction is made possible by the modern tools of gauge theory,  which reveal a simple relational law describing rotational effects. In this way, we can save the phenomena of Newtonian mechanics using conserved charges and relationalism. In a second paper, we will further explore the ontological and explanatory implications of the theory developed here.    
}

\clearpage

\tableofcontents

\section{Introduction}

\subsection{Rotation: the downfall of relationalism?} 
\label{sub:rotation_the_downfall_of_relationalism}


In the Scholium to the Definitions given in the \emph{Principia} \citep{newton1999principia}, Newton introduces and defends his notion of absolute space in terms of real and imagined experiments. The problem Newton set himself was the detection of motion with respect to an immovable space not subject to direct observation:
\begin{quote}
It is indeed a matter of great difficulty to discover, and effectually to
distinguish, the true motions of particular bodies from the apparent;
because the parts of that immovable space, in which those motions
are performed, do by no means come under the observation of our
senses. [p.82]\end{quote}
Despite this ``great difficulty,'' Newton offers a concrete proposal in terms of ``forces'' and ``apparent motions'':
\begin{quote}Yet the thing is not altogether desperate; for we have some arguments
to guide us, partly from the apparent motions, which are the
differences of the true motions; partly from the forces, which are the
causes and effects of the true motions. [p.82]
\end{quote}
Perhaps the most famous experiment used to illustrate the determination of such absolute motion is `Newton's bucket'. In this experiment, a bucket filled with water is spun, and the water is seen to creep up the walls of the bucket as it recedes from the axis of rotation. There is now general consensus among historians  \citep{RYNASIEWICZ1995295,RYNASIEWICZ1995133} that the bucket experiment should be read as a reaction to Descartes' views of motion, and not in the terms of modern debates concerning the nature of space-time. But regardless of the historical context, Newton's rotating bucket still provides a vivid illustration of how the receding water can be used to determine rotational motion with respect to absolute space, and thereby threaten a relationalist view.

That Newton focused attention on examples involving rotation is highly suggestive, and foreshadows the focal point of the subsequent \emph{absolute-relational debate}. The subject of this debate is space-time itself, and concerns its status either as an absolute entity in its own right --- the \emph{absolutist} position --- or as an abstraction inferred from the relations between material bodies --- the \emph{relationalist} position. In \emph{World Enough and Space-time}, John Earman \citep{Earman1989} devotes an entire chapter to the topic of rotation and its role in the absolute-relational debate. There he describes and comments on several failed historical attempts to understand rotation in relational terms:
\begin{quote}
  Those who want to deny that the success of Newton’s theory supports the absolute side of the absolute--relational controversy are obliged to produce the details of a relational theory that does as well as Newton’s in terms of explanation and prediction, or else they must fall back on general instrumentalist arguments.[p.65]
\end{quote}
He then laments that ``there are features of rotation that make it an especially difficult challenge for the relationalist.'' To drive the point home, he comments on the difficulty of providing a purely relational theory of rotation in the light of several prominent inadequate attempts:
\begin{quote}
  Newton, Huygens, Leibniz, Berkeley, Maxwell, Kant, Mach, Poincar\'e --- these are names to conjure with. The fact that not one of them was able to provide a coherent theory of the phenomena of classical rotation is at first blush astonishing.[p.89]
\end{quote}
These failures, Earman contends, are a testament to ``the difficulty of the problem'' that highlight ``the strengths of the preconceptions and confusions about the absolute-relational debate.''

In a more recent review, \citet[Section~\S6.1.1]{pooley_sub-rel} asks whether a relationalist can find dynamical laws in terms of relational quantities that can describe the full set of phenomena of Newtonian mechanics. The answer is rather pessimistic: ``As far as I know, no one has seriously attempted to construct such laws. Even so, one knows that any such laws will exhibit a particularly unattractive feature: they will not be expressible as differential equations that admit an initial value formulation.'' Similarly, in an article directed at `rehabilitating relationalism,' \cite{belot1999rehabilitating} describes how the symplectic reduction methods of \cite{marsden1992lectures}\footnote{ This work is expanded upon in \cite{butterfield2006symplectic} with a discussion of the philosophical implications. See also \cite{abraham1978foundations,marsden1995introduction} for more details. } applied to Newtonian mechanics in the presence of angular momentum ``fails to lead to a Hamiltonian theory on the relative phase space'' because it can only be expressed in terms of ``two additional variables''  that ``have no transparent geometric interpretation for relationalists.''\footnote{ These two variables are directions on a 2-sphere that represent the changing orientation of the angular momentum vector in a frame co-rotating with a rigid body. See \citet[Section~\S2.3.4]{butterfield2006symplectic} for a description of these variables. } He thus concludes that ``Newtonian physics cannot be reformulated as a Hamiltonian theory on relative phase space. Hence, there is, after all, no strict relationalist interpretation of Newtonian theory —-- \emph{rotation is the downfall of relationalism}.'' [Emphasis added.]

This now-orthodox view would seem to put the relationalist in an impossible position with regard to rotational effects due to angular momentum. Yet, to use Newton's words against him: the thing is not altogether desperate. In spite of the prevailing pessimism in the philosophical literature, in this paper we will give a fully coherent theory of classical rotation in relational terms. Our theory builds upon the well-known work of \cite{LittlejohnReinsch}, which  is more flexible than the formulations based on symplectic reduction, and is used extensively in molecular physics, chemistry and engineering for describing the rotation of rigid bodies in terms of relational quantities. Contra \cite{pooley_sub-rel} and \cite{belot1999rehabilitating}, the theory we present is expressible in terms of well-posed differential equations and makes use of geometric quantities that have a straightforward interpretation for the relationalist.

To overcome the concerns expressed above, our theory makes use of an enriched form of relationalism.\footnote{ This enriched relationalism, however, is not of the kind proposed in \citet[Section~\S III.F.2]{sklar1977space} or \citet[Ch.4,Section~\S1]{van1970introduction} (and discussed critically in \cite{maudlin1993buckets} and \cite{huggett1999manifold}), which takes relational accelerations or inertial frames as primitives. See the discussion below. } In particular, it adds an $\SO(3)$ charge that couples to particle relations in ways commonly seen in standard gauge theories such as electromagnetism and the Standard Model of particle physics. Our proposal postulates a new form of relational law involving universal coupling constants representing the $\SO(3)$ charge that can be treated no differently from other constants of nature such as Newton's gravitation constant $G$.\footnote{ In this sense, our approach resembles a proposal along the lines of a `regularity account' advocated by \cite{huggett2006regularity} but in terms of an explicit formulation of the relational law. } This new law can be expressed entirely in terms of relational quantities using a curved gauge connection. In a second paper, we will argue, in response to Earman's challenge, that this theory can match Newton's in terms of explanation and prediction, without having to appeal to instrumentalist arguments. There we explore in greater detail the ontological and explanatory implications of our new theory.

 To understand the novel features of our theory, it is useful to compare it with existing approaches. These approaches can be broadly classified by the way in which they implement different classical space-time symmetries. This is done either directly at the level of the space-time or at the level of the state-space in terms of configuration space or phase space. The basic differences between these approaches can be understood in terms of the amount of structure of the original absolute Newtonian configuration space that is retained in the relational theory. An accounting can be made in terms of the number and nature of the degrees of freedom that are removed from the Newtonian theory. A simple way to do this is to list the number of fixed standards that exist in a theory. For example, a particular neo-Newtonian theory might retain a fixed standard for inertial frames but reject a fixed standard of position. Depending on whether a theory is based on space-time, configuration-space or phase-space, it may be more or less natural to retain or reject a particular standard.

A prominent configuration-space based approach that makes use of  elementary gauge theory methods is that of \cite{Barbour_Bertotti}. This theory is strongly relationalist in that it rejects fixed global standards of position, velocity, orientation, rotation and duration. The theory is then expressed implicitly on a \emph{relative configuration space}, where all absolute standards have been removed. In order to remove the fixed standard of rotation, the theory requires the angular momentum of the universe to be zero. The phase-space-based methods of \cite{marsden1992lectures} (cf. also \cite{abraham1978foundations,marsden1995introduction,butterfield2006symplectic}) suggest that the restriction to zero angular momentum is the only possible way to implement a relational theory of rotation because, in the presence of angular momentum, two additional degrees of freedom appear that do not have any obvious relational interpretation within the phase-space formalism. As noted above, this point was emphasised in \cite{belot1999rehabilitating}. In \cite{pooley2002relationalism}, vanishing angular momentum was then touted as a \emph{prediction} of relationalism highlighting the difficulty of describing angular momentum effects in relational terms.

Similar conclusions have been reached from the point of view of space-time based approaches. Recently, it has been argued by \cite{saunders2013rethinking} and \cite{knox2014newtonian} that Newtonian mechanics should be expressed in terms of Maxwell--Huygens space-times, which \emph{do not} have a fixed standard of linear acceleration but \emph{do} have a fixed standard of rotation.\footnote{ The space-times in question are referred to as `Maxwellian space-times' in \cite{Earman1989} and `Newton-Huygens space-times' in \cite{saunders2013rethinking}. We will use the term `Maxwell--Huygens space-time' introduced in \cite{weatherall2016maxwell}. These differ from neo-Newtonian space-times in that they do not have a fixed global inertial structure.} These author's retention of a standard of rotation (but not orientation) is motivated by the desire to accommodate a universe with non-zero angular momentum.

While we accept that it is possible to describe Newtonian mechanics without introducing fixed standards of linear acceleration and orientation; we deny that, independently of metaphysical considerations, a standard of rotation is strictly necessary. We will defend our position by explicitly giving a theory that is independent of any rotational standard but can nevertheless also account for universal angular momentum. Such a formulation is made possible through the use of the modern advances of gauge theory that were not available to Earman's list of grand ``names to conjure with.''  Importantly, this machinery can accommodate non-zero curvature effects and a more general class of dynamically defined relational frames that allow for a straightforward geometric interpretation of the features of the symplectically reduced theory that \cite{belot1999rehabilitating} deemed awkward.\footnote{ Concretely, the time-dependence of the two orientation variables in a frame co-moving with a rigid body can be understood in terms of the curvature of the rotational bundle resulting in the absence of a horizontal section. See the discussions at the end of Section~\S\ref{sub:gen Lorentz}.}

\subsection{Saving the phenomena with conserved charges} 
\label{sub:saving_phenomena_without_rotation}


Our theory is expressed in terms of two distinct relational force terms: a Lorentz-like force, analogous to the one of electromagnetic theory, describing Coriolis effects; and a force term due to a mass-like quadratic potential describing centrifugal effects. Formally the novel force terms are represented on an $SO(3)$ principal fibre bundle over a base space corresponding to the relative configuration space of \cite{Barbour_Bertotti}. In describing Coriolis and centrifugal effects, our theory saves the phenomena of Newtonian mechanics.  In other words, all rotational phenomena can be described solely in terms of a relational ontology with an additional $SO(3)$ charge. Our theory can thus be viewed as both a more explicit statement of the theory presented in \cite{Barbour_Bertotti}, in that we perform an explicit reduction to a relative configuration space, and a generalisation of it to the case of non-zero angular momentum.

One of the most significant features of the new construction is that it provides an enriched relational interpretation on the rotational bundle when the angular momentum of the system is non-zero. Such an interpretation suggests a new kind of \emph{enriched relational ontology} that combines certain elements of a fully-reduced description with additional elements inspired by the global structure of the bundle.  But because the additional elements correspond to a modification of a relational dynamical law, no new accelerations or inertial frames must be included as \emph{primitives} in the ontology such as those proposed in \citet[Section~\S III.F.2]{sklar1977space} or \citet[Ch.4, Section~\S1]{van1970introduction}. Instead, inertial frames can be defined \emph{dynamically} in terms of relational quantities and the value of the $\SO(3)$ charge (see Section~\S\ref{sub:rot non issue} for a more detailed discussion of this point).

In our picture, we show that all inertial effects caused by rotational phenomena can be described directly in terms of conserved quantities on relative configuration space. This comes about because the motion along the fibres of a curved bundle is conserved if the fibres have a `rigid' geometry.\footnote{ More precisely, by `conserved motion' we mean rigid motion in the sense that the vertical velocities of the dynamical trajectories are conserved between fibres, which lie along Killing directions of the kinematical metric of the theory (see Theorem~\ref{ftnt:Killing} of Section~\S\ref{sub:gen Lorentz}). } This motion can be projected directly onto the relative configuration space where it is described by the Lorentz-like force and mass-like quadratic potential terms mentioned above.

The construction here is intimately linked to that of \cite{LittlejohnReinsch}, where equivalent terms were obtained when working within particular choices of rotational frame. The full principal fibre bundle picture presented here has certain advantages, the most important of which is that it does not refer to frames at all. This is particularly useful for eliminating the awkward frame effects that are seen in attempts \citep{marsden1992lectures} to symplectically reduce the rotations. Our formalism can also be applied to any semi-simple Lie group, and can thus be directly compared to the decompositions used in Kaluza--Klein theory to which we will return shortly. In the special case where the angular momentum is equal to zero, the motion along the fibres vanishes and is thus not required to describe the reduced theory.

While this picture matches the strong relationalism advocated in \cite{Barbour_Bertotti}, there are notable differences between that work and ours. Their framework cannot be applied to isolated subsystems of the universe because such subsystems can have non-zero conserved angular momentum.  The theory of \cite{Barbour_Bertotti} therefore applies only in a cosmological context, and cannot satisfy the concept of \emph{subsystem-recursivity} advocated in \cite{pittphilsci16623} according to which ``interpretative conclusions about a sector of a theory can be deduced from considering subsystems of other models of the same theory.'' Because our theory can accommodate non-zero angular momentum, it has the advantage of being equally applicable to subsystems as to the universe as a whole, and can be made to satisfy subsystem-recursivity (see the end of Section~\S\ref{sec:LR_intro}).\footnote{Gluing back reduced systems however comes with its own complications (cf. \cite{RovelliGauge2013, GomesStudies} and \cite[Sec. 3]{DES_gf}). } Thus, our theory is more flexible and fits a broader class of contexts for discussing symmetry\footnote{ Including many of the contexts discussed in \cite{pittphilsci16622}.} regardless of the status of the angular momentum of the universe.

Our construction further allows us to pinpoint the precise technical difference between the mathematical representation of the dynamics of rotation and translation. This difference lies in the nature of the gauge-connection: it is necessarily curved for rotations but flat for translations. This difference leads to many simplifications in the translational case.

To better understand the flexibility of our enriched ontology, it is helpful to consider the similarities and differences between our treatment of rotations and the treatment of gauge symmetries in Kaluza--Klein theory. Standard Kaluza--Klein theory makes use of the kind of enriched ontology introduced in this paper by positing particles that move in a fixed, symmetric and extra-dimensional background geometry. The explanatory role played by this ontology in Kaluza--Klein theory is, however, very different to the rotational case.

In Kaluza--Klein theory, the gauge orbits take on more ontological significance than in Yang--Mills gauge theories because they are interpreted as extra space-time dimensions. But the orbits do not have the full representational capacity of ordinary directions in space-time because motion along these directions is rigidly constrained to accommodate a projection of the motion onto the base space. Some argument must then be supplied to specify the condition for projectability.

The justification in Kaluza--Klein theory for the projection is epistemic: the projected directions are too small for motion along them to be empirically accessible. The phenomena that we do have epistemic access to are then sufficiently well-described by the projection of the constrained motion in the extra dimensions to the familiar $(3+1)$-dimensional motion on the base space, where it is interpreted as motion in space-time in the presence of background fields. The constraints on the motion translate to conserved charges of the corresponding symmetry group. Because of the epistemic interpretation of the projection, the extra dimensions play a prominent role in the possible empirical implications of the theory.

This can be compared to the epistemic status of the orbits of spatial rotations in Newtonian mechanics. In this case, rotated and translated representations of the universe are empirically indiscernible even in principle.\footnote{ Here we are referring to `kinematical shifts' of all the particles in the system and not the more cryptically defined `dynamical shifts' that, in this case, would involve changing the angular momentum of the system. } The argument for the projection can then be made in terms of the Principle of the Identity of Indiscernibles. For translations this can be done without introducing an enriched ontology following, for example, the arguments in \cite{saunders2013rethinking} or the explicit reduction performed in \cite{LittlejohnReinsch} that is summarised in Section~\S\ref{sub:intro LR}. But for rotations, a simple construction is not possible and an enriched ontology is necessary. As we will see, there is a close formal analogy between this and the Kaluza--Klein case, where the relevant conserved charge for the rotation group is the angular momentum of the universe. The differences between the way the projection is interpreted --- in the rotational case as an ontological identification and in the Kaluza--Klein case as an epistemic constraint --- thus illustrates the flexibility and generality of the new kind of enriched ontology introduced here.

In this paper we will focus mainly on the formal representations of our proposed enriched relational ontologies. Questions surrounding the interpretation of the projection to base space, the ontological and explanatory implications of rigid motion, and the comparison with Kaluza--Klein theory will be taken up in our second paper. The arguments of this paper will therefore necessarily be rather technical in nature. That's unsurprising. After all, the main advantage that this construction has over the many previous failed attempts to produce a relational theory of rotation is that it makes use of the modern mathematical advances in gauge theory. But we hope these technicalities will not be seen as impediments to conceptual clarity. Our mathematical labours will bear significant conceptual fruit: a compact description of a relational theory of classical (i.e., non-relativistic) rotation in terms of an enticing new proposal for formulating the ontology of such a theory.

\subsection{Prospectus}

In Section~\S\ref{sec:the_angular_momentum_of_the_universe}, we build intuition for our theory by comparing it to existing theories in the literature. For convenience, the technical results of the paper are summarised in Section~\S\ref{sec:tech summary}. The principal fibre bundle formalism is then introduced in Section~\S\ref{sec:PFB_intro}, where we use it to define the horizontal projection of curves in a fibre bundle in terms of the basic structures required for our construction: parallel transport (via a connection-form) and curvature. We then develop our theory in Section~\S\ref{sec:KK_st}, Section~\S\ref{sec:KK_config} and Section~\S\ref{sec:LR}.  In Section~\S\ref{sec:KK_st} we review the standard formulation of Kaluza--Klein theory for semi-simple Lie groups over space-time. Given a background connection-form, a metric is defined on the bundle in Section~\S\ref{sub:bundle metric} that is natural and respects the bundle symmetries. Using this metric, it is possible to derive the standard Lorentz force as the difference between the acceleration along geodesics in space-time and the projection to space-time of geodesics in a higher-dimensional bundle. The decomposition used in Section~\S\ref{sec:KK_st} then provides the tools to build a formal analogy to the rotational case. In Section~\S\ref{sec:KK_config} we switch gears and develop a principal fibre bundle formalism for configuration space. We turn around the logic of Kaluza--Klein theory and use the kinematical metric and its symmetries to uniquely select a connection-form (Section~\S\ref{sec:config_varpi}) and a bundle curvature form (Section~\S\ref{sec:curvature}) on configuration space. Finally, in Section~\S\ref{sec:LR} we apply the general construction of Section~\S\ref{sec:KK_config} to the case of translations and rotations in a Newtonian $N$-particle system. In Section~\S\ref{sec:LR_intro}, the connection and curvature forms for this case are explicitly constructed. This leads to the main result (Equation~\ref{config_Lorentz}), derived in Section~\S\ref{sec:KK_proj}, that gives a representation of the dynamics on the rotation-free reduced space in terms of a Lorentz-like force and a mass-like quadratic potential. Section~\S\ref{sec:conclusions} concludes. The appendices provide more detailed proofs of some of the technical requirements of Sections~\S\ref{sec:PFB_intro} and \S\ref{sec:KK_st}.

\section{The angular momentum of the universe} 
\label{sec:the_angular_momentum_of_the_universe}

\subsection{The (non-)issue with rotations} 
\label{sub:rot non issue}

Let us first consider a universe with zero angular momentum. In this case, there is general agreement \citep{belot1999rehabilitating,pooley2002relationalism} that a relational theory of rotation can be broadly constructed along the lines suggested by \cite{Barbour_Bertotti}. Indeed, \cite{pooley2002relationalism} consider vanishing angular momentum to be a real-world prediction of a fully relational theory. Zero angular momentum is further claimed \citep{belot1999rehabilitating} to be a ``contingent fact'' supported by the best available evidence from cosmology (e.g., due to the observed isotropy of the CMB).

Such claims are, in our opinion, implausible since they inappropriately extend the analogies between Newtonian and general relativistic space-times. In general relativity there are inherent difficulties in defining the angular momentum of the universe. In a spatially closed universe, there is not even any known way to define global angular momentum. While definitions of angular momentum are available for open universes, these definitions rely on fixed asymptotic boundary conditions, and it is not clear how to construct covariant boundary conditions that would simultaneously: allow for universal angular momentum, be compatible with cosmological observations, and agree with relationalist intuitions. There is therefore no rigorous sense in which vanishing angular momentum has been established as a contingent fact of our universe; and thus there is no good reason to reject the possibility of universal angular momentum --- particularly in a Newtonian setting.

In the presence of angular momentum, all parties to the absolute-relational debate agree that the evolution of a system of point particles is underdetermined by their initial relative positions and momenta. One extra vectorial datum, expressible as the value of the total angular momentum, is required to determine the evolution of the system. The need for such an extra vectorial datum was recognised long ago by \citet{Tait_shape} and \citet{Lange_shape} in the context of the initial value problem for force-free Newtonian point particles in three dimensions and was later emphasised by \cite{Poincare_sci_hypo} and eventually by Barbour (e.g., \citeyear{Barbour_Mach,Barbour_Bertotti}). Such a requirement can plausibly be taken to motivate accounts of angular momentum in terms of Maxwell--Huygens space-time (as in for example \cite{Earman1989,saunders2013rethinking,knox2014newtonian}), which retains a fixed standard of rotation; or in terms of neo-Newtonian space-time, which further retains a fixed inertial structure.  Alternatively, the extra data can be taken to be primitive relational accelerations along the lines of \citet[Section~\S III.F.2]{sklar1977space} or inertial frames along the lines of \citet[Ch.4, Section~\S1]{van1970introduction}. The account we will present here, however, illustrates that the retention of a standard of rotation or the positing of primitive accelerations or inertial frames is not necessary in order to allow for non-zero angular momentum.

Upon the symplectic reduction presented in \cite{marsden1992lectures}, the angular momentum of the original system decomposes into a constant magnitude and two time-dependent orientations. The orientations represent the time varying directions of the relational angular momentum in a frame co-rotating with the system \citep[Section~\S2.3.4]{butterfield2006symplectic}. These directions are said to ``have no transparent geometric interpretation for relationalists'' \citep{belot1999rehabilitating};  although their evolution can be determined autonomously in terms of initial data on relative configuration space and a specification of the value of the conserved angular momentum.

The difficulty of finding a co-rotating frame where the angular momentum is conserved is related to a more general problem of working with frames that only make use only of instantaneous configurations and no further dynamical information. The lesson from the Scholium to Newton's Principia and from the Lagrangian-based approach of \cite{Barbour_Bertotti} is that dynamical considerations are necessary to determine the privileged frames in which angular momentum is conserved. This information cannot be extracted from the instantaneous configurations alone but requires knowledge of the theory's laws.

These facts are nicely encoded in the fibre bundle formalism. For rotations there is in general  no simple function of the \emph{instantaneous} relative configurations alone that can be used to define frames where angular momentum is conserved because, as we will see, the rotational bundle is curved.\footnote{ Formally, this is because curvature prohibits the existence of a horizontal section on the fibre bundle. } This fact is exemplified, both in \cite{LittlejohnReinsch} and in the symplectic reduction formalism of \cite{marsden1992lectures}, by the non-conservation of angular momentum in the frames used in those papers. However, using the principal fibre bundle formalism, it is always possible to find an \emph{anholonomic frame} associated with the \emph{extremal curves} (i.e., the dynamically possible models of the theory) on the bundle where the angular momentum is indeed conserved. To determine such frames one requires knowledge of the theory's laws, which determine the extremal curves, and the value of the $\SO(3)$ charge. Such a construction gives a dynamical definition of inertial frames that is along the lines of \cite{Barbour_Bertotti} but also works in the presence of non-zero angular momentum.

The most interesting aspect of our new theory, therefore, is its ability to represent angular momentum directly on relative configuration space as a charge no different from that of a classical electrically charged particle moving in a background electromagnetic field. From a modern perspective, the motion of a charged particle has two independent interpretations: one in which a charge needs to be posited and another where it does not. In the former standard picture, the dynamical evolution of the particle in space-time is only determined after adding one datum: the charge of the particle which is to be acted on by the background field. In the latter picture, a Kaluza--Klein interpretation of the same charged-particle dynamics eschews this extra datum. In this picture's enriched ontology, forced motion is entirely geometrised as the observable ``shadow'' of \emph{free} motion in a space-time with one additional dimension. Here the charge datum is re-described as a conserved momentum along the extra dimension. In this analogy, the account in terms of a Maxwell--Huygens space-time is analogous to a Kaluza--Klein interpretation where the extra dimension is taken to be real in light of its explanatory role in accounting for charged motion. While such a realist view of the extra dimension is always available, it is clearly not necessary and is even non-standard. The relative merits of an absolutist versus an enriched relationalist account will be analysed in more detail in our second paper.

\subsection{The two-body system} 
\label{sub:the_two_body_system}

In this section we attempt to build intuition for how an extra vectorial datum can, on the one hand, be interpreted as a dynamical degree of freedom due to angular momentum and, on the other hand, be interpreted in relational terms as a charged coupling constant. To do this, we consider the simple example of a two-body system rotating under some potential. Because this system can be confined to a two dimensional plane, a projection of the rotational motion will leave a one-dimensional system in which curvature must be equal to zero. Vanishing curvature implies that the Lorentz-like force term of \eqref{config_Lorentz}  that we will derive in Section~\S\ref{sec:KK_proj} is zero. Nevertheless, the mass-like quadratic potential is non-zero, and so this example will serve to illustrate this effect in an intuitive setting.

To define the model, we use a coordinate system in the plane of rotation and set the origin of the coordinate system to be the location of one of the particles. If $r$ is the distance between the particles and $\theta$ is an angle parametrising the orientation of the system in absolute space, then the Lagrangian for the system is
\begin{equation}
  \mathcal L = \frac m 2 \left( \dot r^2 + r^2 \dot \theta^2 \right) - V(r)\,,
\end{equation}
where $m$ is the mass of the second particle.\footnote{The kinetic energy of the first particle is zero in these coordinates.} Because $\theta$ is a cyclic variable, it can be eliminated via a Routhian reduction by integrating out its Euler--Lagrange equations. This introduces the constant of motion $L = m r^2 \dot \theta$, which is the angular momentum of the coplanar system. In terms of this constant, the reduced theory can be shown to have the reduced Lagrangian
\begin{equation}
  \mathcal L_\text{red} = \frac{m \dot r^2}{2} - \frac{L^2}{2 m r^2} - V(r)\,.
\end{equation}
Removing the angular variable $\theta$ thus introduces the well-known centrifugal potential $V_\text{eff} = \frac{L^2}{2 m r^2}$. This potential is the two-dimensional analogue of the mass-like term of our main result \eqref{config_Lorentz}, which is quadratic in the angular momentum $L$.

The reduced theory is empirically equivalent to the original theory defined in absolute space. Because it is expressible entirely in terms of $r$, no standard of rotation need be introduced. The theory is also explanatory: all rotational effects are accounted for by a central repulsive $1/r^3$ force with a coupling that can be determined empirically. Any apparent underdetermination in the system can be resolved by specifying the value of this coupling in a manner no different to the specification of Newton's constant $G$ in a gravitational system.

The simplicity of this two dimensional example, however, masks the complexities of the three dimensional theory. In three dimensions, the curvature of the rotational bundle is no longer zero, and an understanding of these effects requires the tools of modern gauge theory. We now turn attention to these.

\subsection{The rotational bundle}\label{sub:intro LR}

In their groundbreaking work, \citet{LittlejohnReinsch} carry out the explicit elimination of those degrees of freedom that can be eliminated on the basis of translational and rotational invariance by working with particular standards of rotation; i.e.,  particular choices of frame. The bundle construction for the translations is largely uninteresting and can be rather explicitly reduced. After doing this, they construct a rotational bundle over relative configuration space. They relate the properties of the gauge-potential (i.e. the connection-form in a particular choice of frame) and of the gauge curvature  to the classical dynamics of point particles. These gauge fields have a simple physical interpretation, which can be understood in terms of elementary ideas about conservation of angular momentum, and the rotations generated by deformable bodies with changing moments of inertia.\footnote{This was the approach taken by \cite{guichardet1984rotation}, who dealt with the kinematics of deformable bodies such as molecules and falling cats. The same discovery was made independently by Shapere and Wilczek (\citeyear{wilczek1989a,shapere1989b}), evidently as a by-product of their more substantial work on the gauge theory of the locomotion of objects such as microorganisms in a viscosity-dominated medium \citep{shapere1987self,shapere1989c,shapere1989d}.} \citet[p.215]{LittlejohnReinsch} write: ``What is particularly remarkable about these developments is the manner in which the entire structure of nontrivial connections on non-Abelian fibre bundles emerges from elementary mechanical considerations.''

The theory that emerges is not only mathematically elegant, but has also been shown to be immensely useful in the study of molecular dynamics, where one is mostly interested in the properties of the reduced dynamics. At the time of writing of \citep{LittlejohnReinsch}, the notion that the internal dynamics of $N$-body systems is a gauge theory was a new one in the literature of applied physics, chemistry, and engineering. But it has since been proven to have wide-ranging implications for the understanding of such systems.

Most importantly for us, the work of \citet{LittlejohnReinsch} unifies the treatment of rotations with that of other symmetries in the Newtonian framework. In so doing, it clarifies the distinctions between the previous historical treatments of the translation group and those of the rotation group. The most significant way in which the current work differs from the formulation of \citet{LittlejohnReinsch} is that they specialise to specific choices of rotational frames. In our formalism, which recasts the reduction in the language of principal fibre bundles, such choices are never necessary. The geometric picture we obtain is both technically slender and conceptually insightful since it emphasises the construction is independent of any choice of rotational standard. Moreover, it enables us to articulate the formal analogy to electromagnetism and Yang--Mills theory using the Kaluza--Klein formalism. These differences will play a central role in the explanatory considerations of the second paper. For a more technically detailed commentary on the differences between these approaches, see Section~\S\ref{sec:LR}.


\section{Summary of the technical results of this work} 
\label{sec:tech summary}

In this section, we summarise for convenience the technical results of the theory developed in Section~\S\ref{sec:KK_st}, Section~\S\ref{sec:KK_config} and Section~\S\ref{sec:LR}. To begin, consider that the standard Lagrangian for a conservative central force system is invariant only under time-\textit{independent} changes of orientation. One of the primary technical achievements of this paper is to show that the gauge formalism is powerful enough to extend this symmetry and thereby project the dynamics, given one added choice of constant vector, to a reduced configuration space of inter-particle relations conforming to the intuitions reported in Section~\S\ref{sec:the_angular_momentum_of_the_universe}.
 
Consider the case where the potential energy function is velocity-independent and the kinetic and the potential terms are each invariant under some time-independent transformation. Moreover, the kinetic term is given by the norm squared of the velocity, and this norm is induced by a symmetry-invariant inner product, which we will call the \emph{kinematical metric}. Using these properties of the kinematical metric and the known symmetry orbit induced by the action of the gauge group on configuration space, it is possible  to treat the configuration space as a principal fibre bundle for this symmetry group and to construct a dynamical connection 1-form on this bundle known as an \emph{Ehresmann connection}.\footnote{This is formally similar to the best-matching construction of Barbour (cf. \citet{Barbour_Bertotti} and \citet[Ch. II.5]{Flavio_tutorial} and references therein). See  also section \ref{sec:LR_intro}. }

Using the notion of orthogonality implied by the kinematical metric, we obtain a kinetic term in the Lagrangian that decomposes into one term that admits only velocities orthogonal to the orbit and another that admits only velocities parallel to the orbit. The orthogonal part, by the properties of an Ehresmann connection, is fully invariant. The lack of invariance of the theory is therefore entirely contained in the contribution parallel to the orbit. Because the orbits are by assumption Killing directions of the kinematical metric, the parallel velocities will be conserved along dynamical trajectories (see Theorem~\ref{ftnt:Killing} of Section~\S\ref{sub:gen Lorentz} below). In the case of rotations, such velocities map to the angular momentum of the system: the extra vectorial datum, discussed in Section~\S\ref{sub:rot non issue}, needed to define a complete projection of the theory to the reduced space.

A second technical achievement of this paper requires us to review the Kaluza--Klein construction for non-Abelian Yang-Mills theories. A comparison with that framework shows precisely how inertial motion on the full bundle, when projected to the base space, differs from inertial motion on the base space. Most interestingly, the new structures admit a compelling geometric interpretation in terms of a Lorentz-like force term with the Ehresmann connection generating the relevant curvature contribution. Using this approach, we will be able to offer a geometric interpretation and a classification of the symmetry reductions of different systems. In particular, our interpretation classifies the differences between quotienting with respect to: (i) translations, (ii) rotations, and more generally (iii)  any semi-simple Lie group.
 

One final point that our formalism will help to clarify is the representational difference between rotation and translation. As we will see, reduction by rotations is more involved than reduction by translations. It is also slightly more involved than the reduction process in the standard Kaluza--Klein formulation of electromagnetism.

Translations, it turns out, are remarkably simple to represent owing to \textit{three} independent mathematical facts. The first two facts are that the group action of translations on configuration space and the relevant kinematical metric along the orbits are both independent of the configuration itself. The third relevant fact is that the bundle curvature vanishes for translations. In the Kaluza--Klein construction for electromagnetism, the group action and the kinematical metric along the group orbit are both configuration-independent but the bundle curvature is non-zero. In the Kaluza--Klein construction for general semi-simple Lie-groups, only the kinematical metric along the orbit is configuration-independent: the curvature of the relevant bundle is non-zero \emph{and} the action of the group is not configuration-independent. And in the case of rotations, \emph{none} of these three conditions are satisfied. In particular, the dependence on configurations of the kinematical metric along the orbits of the configuration space induces the additional mass-like quadratic potential term of main result, equation~\eqref{config_Lorentz} of Section~\S\ref{sec:KK_proj}, in the reduced description in addition to the one due to the Lorentz-like force.

\section{A brief introduction to fibre bundles}\label{sec:PFB_intro}

The modern mathematical formalism of gauge theories relies on the theory of principal  (and associated) fibre bundles. We will not give a comprehensive account here (e.g. \citet{kobayashivol1}), but only introduce the necessary ideas and objects. In Section~\S\ref{sub:PFB}, after giving the definition of a principal fibre bundle, we derive a fully gauge-covariant formula, \eqref{eq:hor_proj_v}, for the horizontal projection of curves in such a bundle. This projection defines a notion of parallel transport in terms of a connection-form and its curvature \eqref{eq:curv_h2}, which we will make extensive use of in subsequent sections.

A simple example of a principal fibre bundle is as follows. Given an $n$-dimensional manifold $M$, thought of as representing space-time (though we will not explicitly need any non-trivial metric structure of space-time), the space of linear frames over $M$ is a principal fibre bundle with structure group $GL(n)$ taking $M$ as the base space. Each element of the ``fibre'' over each point of the base space $M$ consists of a linear frame, i.e. a basis, of the tangent space to $M$  at that point. In this example, it is clear that there is no ``zero'' or identity element on each fibre. But there is a one-to-one map between the group $GL(n)$ and the fibre: we can use the group to go from any frame over that point to any other.  

The main idea underlying the physical significance of the internal space in a fibre bundle is  perhaps best  summarized in the groundbreaking original paper by  \citep{YangMills}. They write: 
\begin{quote} The conservation of isotopic spin is identical with the requirement of invariance of all interactions under
isotopic spin rotation. This means that when electromagnetic interactions can be neglected, as we shall hereafter assume to be the case, the orientation of the
isotopic spin is of no physical significance. The differentiation between a neutron and a proton is then a
purely arbitrary process. \end{quote}

The limitations on how to identify ``a proton'' at two different points of space-time are imposed by a connection-form: which is another structure on the bundle. 
That is, a connection-form $\omega$ allows us to define which points of neighbouring fibres are taken as equivalent to an arbitrary starting-off point in an initial fibre. 
In the example of linear frames, it gives us a notion of  ``parallel transport'' of the basis as we go from an initial choice over one point of $M$, to a neighbouring fibre. Curvature then acquires meaning as non-holonomicity. That is, starting from a given element of a given fibre and following such an identification of frames along different  paths in the base space $M$, and arriving back at  the same fibre, the points at which one arrives---i.e. the final  elements  on the bundle obtained by this process---may still differ by a gauge transformation, i.e. a transformation along the fibers like an element of $GL(n)$ in the frame bundle. In most physical applications of principal fiber bundles, it is this disagreement that carries physical consequences. In the Yang and Mills quote, if you and I started with  protons at point $x$, and you stayed put with yours while I parallel transport mine  around a loop within spacetime $M$, I might come back with what you would call a `neutron'. 

 There are two consistency conditions that  a connection-form must satisfy in order to provide such a standard of identity. First, the parallel transport to a neighbouring fibre should commute with the group action; i.e. there is a sense in which it does not really depend on what we choose as the starting point.  Equivalently, there is an equivariance property that $\omega$ must satisfy. Secondly, there must be exactly one choice of parallel transported frame per direction of $M$. All the relevant properties of gauge transformations can be derived from these two. 
 
 We are now going to formalize this intuitive description.

\subsection{Principal fibre bundles}\label{sub:PFB}

A principal fibre bundle is a smooth manifold $P$ that admits a smooth action of a  (path-connected, semisimple) Lie group, $G$; i.e.,  $G\times P\rightarrow P$ with $(g,p)\mapsto g\cdot p$ for some action $\cdot$ and such that for each $p\in{P}$, the isotropy group is the identity (i.e., $G_p:=\{g\in{G} ~|~ g\cdot p=p\}=\{e\}$). 
Naturally, we construct a projection  $\pi:P\rightarrow{M}$, from $P$ to the set $M$ of equivalence classes given by $p\sim{q}\Leftrightarrow{p=g\cdot{q}}$ for some $g\in{G}$. So the base space $M$ is the orbit space of $P$, $M=P/G$, with the quotient topology; i.e., characterized by an open and continuous $\pi$. By definition, $G$ acts transitively on each fibre. 

Locally over $M$,  it must be possible to choose a smooth embedding of the group identity into  the fibres. That is, for $U\subset M$, there is a smooth map $\sigma: U\rightarrow P$ such that $P$ is locally of the form $U\times G$, $U\subset {M}$; i.e., there is a diffeomorphism $U\times G\to \pi^{-1}(U)$ given by $(x, g)\mapsto g\cdot \sigma(x)$.\footnote{Given $p$, the inverse map is a bit more complicated because we must  find $g'$ such that $g'\cdot p=\sigma(x)$, for some $x$. It will depend on the form of $\sigma$. }  The maps $\sigma$ are called {\it local sections} of $P$.

On $P$, we consider an Ehresmann connection $\omega$, which is a 1-form on $P$, valued in the Lie algebra $\mathfrak{g}$ and satisfying appropriate compatibility properties with respect to the fibre structure and the group action of $G$ on $P$. The Ehresmann connection is the basic object defining a generalized version of parallel transport; i.e., horizontal projection. 

Given the Lie-algebra $\mathfrak{g}$, we define the vertical space $\V_p$ at a point $p\in P$, as the linear span of vectors of the form $\iota_{p}(\xi)$ for $\xi\in \fg$ and $\iota_p:\fg\ra T_pP$ defined as:
\be\label{eq:iota}\iota_{p}(\xi):=\frac{d}{dt}{}_{|t=0}(\exp(t\xi)\cdot p)\,.\ee  Thus $\iota_p(\xi)$ is the tangent vector of the curve through $p$ generated by $\xi$. After defining vertical spaces,  $\iota_p$ is then seen as a linear operator $\fg\ra \V_p$. And then the conditions on $\omega$ are:
\be\label{eq:connection_prop}
\omega(\iota(\xi))=\xi
\qquad\text{and}\qquad
g^*\omega=g^{-1}\omega g=\Ad_g\omega,
\ee
where, for a vector $v\in T_pP$, the pull-back is defined by $g^*\omega_p(v):=\omega_{g\cdot p}(g_* v)$. Furthermove, $g_*$ is the push-forward of the tangent space by, or the linearization of, the map $g\cdot:P\rightarrow P: p\mapsto g\cdot p$, and $\Ad:G\times\mathfrak{g}\rightarrow \mathfrak{g}: (g, \xi)\mapsto g^{-1}\xi g$. 
  
A choice of connection is equivalent to a choice of covariant `horizontal' complement to the vertical space; i.e., $\HH_p\oplus \V_p=T_pP$, with $\HH$ compatible with the group action: $g_*\HH_p=\HH_{g\cdot p}$. Given a direct sum decomposition, the connection-form is:
  \be\label{eq:connection} \omega(\cdot)=\iota^{-1}(\hat V(\cdot))
  \ee
where $\hat V$ is the linear projection  onto the vertical spaces, $\hat V:T_pP\rightarrow \V_p$ and ker($\omega)=\HH_p$. It is easy to see that \eqref{eq:connection} satisfies \eqref{eq:connection_prop} (cf. appendix \ref{app:connection}).

A connection therefore allows us to locally define ``horizontal complements'' to the fibres in $P$. Through such complements one can horizontally lift paths $\gamma$ lying  in $M$ to $P$. These horizontally lifted paths are commonly referred to as ``parallel transports" in $P$ along $\gamma$ with respect to (horizontality as defined by) $\omega$.
   As in the example above of linear frames, when you go around a closed curve  in $M$, parallel transport upstairs in $P$ may land you at a different point on the initial fibre from where you started; e.g., assuming you started from $p$, you may end at $p'=g\cdot p$.   The relation between $p$ and $p'$ (i.e.,  $g$) is called the holonomy of $\omega$ along the closed path $\gamma$. Its infinitesimal analogue is the curvature of $\omega$, 
   \be \Omega=\d_{\text{\tiny{P}}} \omega+\omega\wedge_{\text{\tiny{P}}} \omega\,,\label{eq:curv_h1}
   \ee
   where $\d_{\text{\tiny{P}}}$ is the exterior derivative on the smooth manifold $P$, and $\wedge_{\text{\tiny{P}}}$ is the exterior product on $\Lambda(P)$ (it gives anti-symmetrized tensor products of differential forms).  We can also write the curvature 2-form as (cf. \citet{kobayashivol1}): 
   \be\label{eq:curv_h}
  \Omega=\d _{\text{\tiny{P}}}\omega\circ\hat{H} 
   \ee
   or even as:
   \be
   \Omega(\cdot, \cdot)=\iota^{-1}\hat{V}(\lbr\hat{H}(\cdot), \hat{H}(\cdot)\rbr)\label{eq:curv_h2}\,,\ee
   where the double-square bracket is the commutator of vector fields in $P$. The domain of the differential calculus will be mostly clear from context, and we will therefore only reinsert the subscript, $P$ or $M$ when necessary.

Most important for us will be the covariance properties of the horizontal projection. To get to these, we start by using the two properties \eqref{eq:connection_prop} to show that for a 1-parameter group of transformations $g_t=g_0 \exp(t\xi)$, and for a curve $\gamma_t: I\rightarrow P$, 
 we have: 
$$\varpi(\frac{\d (g_t\cdot \gamma_t)}{\d t}_{t=0})=\varpi(g_0{}_*\left(\frac{\d (\exp(t\xi)\gamma_0)}{\d t}+\dot\gamma\right))= \Ad_{g_0}(\varpi(\dot \gamma)+\varpi(\iota(\xi)))=\Ad_{g_0}(\varpi(\dot\gamma)+\xi),$$
and therefore, defining a horizontal projection $\hat{H}:=\bb 1 -\hat{V}$, we have, from \eqref{eq:connection_prop} (and by the properties of the $\iota$ map, which imply $\iota(\Ad_{g_0}\xi)=g_0{}_*\iota(\xi) $, see \eqref{eq:iota_equi}):
  \begin{eqnarray}
  \hat H_{g_0\cdot p}(\frac{\d (g_t\cdot \gamma_t)}{\d t}_{|t=0})=g_0{}_*\left(\iota(\xi)+\dot\gamma\right)-\iota(\Ad_{g_0}(\varpi(\dot\gamma)+\xi)),\nonumber\\
  =g_0{}_{*}(\dot\gamma-\iota(\varpi(\dot\gamma)))=g_0{}_{*}\hat \HH_p(\dot \gamma)=\hat H_{g_0\cdot p}(g_0{}_{*}(\dot\gamma))\,.\label{eq:hor_proj_v}\end{eqnarray}
  This shows that the horizontal projection of the velocities is fully  gauge-covariant; i.e., even under time-dependent gauge-transformations the time derivative of the transformation, $\xi$, does not appear in the final result. The gauge-covariance of this projection will play an integral role in the analysis to follow. It will be the gauge-connection that will determine a standard of rotation along the curve; a standard that exists whether the angular momentum, or equivalently the vertical motion, vanishes or not.

\section{Kaluza--Klein for space-time}\label{sec:KK_st}

By 1919, Maxwell's electromagnetic theory was well established. Einstein had only
 recently formulated general relativity and it was natural to seek an account of electromagnetism  within the 
theoretical framework of his new theory; that is, through the geometry of space-time. Theodor Kaluza
 achieved this unification by postulating an extra dimension.
Although attractive, Kaluza's idea had two serious defects: the dependence on the fifth
coordinate was suppressed for no apparent reason, and a fifth dimension had never
been observed. These two criticisms were addressed by Oscar Klein, who postulated
a circular topology and a rigid, homogeneous geometry for the fifth dimension. He showed that if the radius was small enough it was possible to keep the dependence on the fifth coordinate, and still justify unobservability and preserve Kaluza's results.

Our present understanding of gauge theory is fully field-theoretic, unlike the particle-field picture used here. In the fully field-theoretic version, there is no particle trajectory, and the field content of the theory, including its Lagrangian, is fully gauge-invariant.\footnote{Discussions on other uses of Kaluza--Klein theory and reduction, e.g. in string theory, go well beyond the scope of this paper.} Nonetheless,  the particle-field idealisation is sufficient for establishing the formal analogy we make use of in this paper.

In the remainder of this section, we review the basic mathematical constructions of standard Kaluza--Klein theory. We begin in Section~\S\ref{sub:bundle metric} by defining the relevant fibre bundle and equipping it with a metric structure that respects the bundle symmetries. In Section~\S\ref{sec:KK_curv}, we compute and compare two importantly distinct notions of curvature: the bundle curvature defined in \S\ref{sec:PFB_intro} using \eqref{eq:curv_h} and the Levi--Civita connection of the metric we define in Section~\S\ref{sub:bundle metric}. The relationship between these two quantities is absolutely crucial to the Kaluza--Klein analysis because it allows us to compute the difference between the projection of geodesic curves on the bundle and geodesics on the base space. This difference, as illustrated by equation~\ref{eq:KK_Lorentz}, is precisely encoded in the usual Lorentz-force term of electromagnetism.

\subsection{The bundle metric} \label{sub:bundle metric}

Throughout this section, $P$ will be a $G$-bundle over the smooth (pseudo)riemannian $m$-dimensional manifold $M$.  $\{X_i\}_{i=1}^m$ will be an orthonormal reference frame over the open set $U\subset M$; i.e., each $X$ is a vector field in $U$ such that, at each $x\in M$, $\{X_i(x)\}_{i=1}^m$ is an orthonormal basis of $T_xM$  and $\{\lambda^i\}_{i=1}^m$ is the associated co-frame such that $\lambda^i*X_j=\delta^i_j$.

$G$ is a $k$-dimensional Lie group, with $\fg $ its Lie-algebra, endowed with an $\Ad$-invariant inner product $K$.\footnote{This inner product always exist for semi-simple Lie groups, usually it is just written as the trace, more generally, it is known as the \textit{Killing form}.  }  This inner product induces an inner product on the vertical spaces via: 
\be\label{kil} \langle \iota_p(\xi), \iota_p(\xi')\rangle_p:= K(\xi, \xi')\,.
\ee
Due to the covariance properties of the $\iota$ map (see \eqref{eq:iota_equi}), this inner product is $G$-invariant:
\be\label{eq:K_inv} \langle \iota_{g\cdot p}(\xi), \iota_{g\cdot p}(\xi')\rangle_{g\cdot p}= K(\Ad_g\xi, \Ad_g\xi')=K(\xi, \xi')= \langle \iota_p(\xi), \iota_p(\xi')\rangle_p
.\ee

Writing the dimension of $P$ as $n$, the Lie-algebra basis is given by $\{\tilde{e}_\sigma\}_{\sigma=m+1}^n$, where dim$(\fg)=k=n-m$, and the dual basis, for $\fg^*$, is $\{e^\sigma\}_{\sigma=m+1}^n$. In general, Roman indices will vary from $1$ to $m$, and Greek indices from $m+1$ to $n$. We will also denote  ${C^\gamma_\sigma}_\beta$ as the structure constants of $\fg$:
$$[\tilde{e}_\sigma,\tilde{e}_\beta]={C^\gamma_\sigma}_\beta \tilde{e}_\gamma\,,$$
where  ${C^\gamma_\sigma}_\beta=-{C^\gamma_\beta}_\sigma $ and  ${C^\gamma_\sigma}_\beta =-{C^\beta_\sigma}_\gamma$. 

Finally, given a metric $h$ in $M$ and a connection $\omega$ --- or the equivalent choice of horizontal space $\HH$ --- we obtain a unique $G$-invariant metric, $\bb G$, on $P$ such that: the decomposition $TP=\HH\oplus\V$ is orthogonal, $\pi_*{}|_{\HH}$ is an isometry, and the vertical inner product is induced by $K$; i.e.,
\be\label{eq:bundle_metric} \bb G:=\pi^*h+K\circ\omega\,. \ee
For $u,v\in\V_p$, we have:
 \be\bb G_p(u,v)=K(\iota_p^{-1}(u),\iota_p^{-1}(v))=\langle u, v\rangle_p\label{eq:vert_ip}\ee
because $\pi_p{}_*(v)=\pi_p{}_*(u)=0$. For $u,v\in\HH_p$ we have
\be\bb G_p(u,v)=h_{\pi(p)}(\pi_p{}_*(u),\pi_p{}_*(v))\ee 
since $\omega_p(u)=\omega_p(v)=0$, and, finally, if  $u\in\HH_p$ and $v\in\V_p$ then $\bb G_p(u,v)=0$ since $\omega_p(u)=0$ and $\pi_p{}_*(u)=0$. Because both $\pi^*h$ and $K\circ\omega$ are $G$-invariant (cf. appendix \ref{app:connection}), $\bb G$ is $G$-invariant. In other words: vertical directions are Killing directions for $\bb G$.  

For ease of notation and using $\{X_i(x)\}_{i=1}^m$ and $\{\lambda^i\}_{i=1}^m$ the set of dual orthonormal frames and co-frames defined above, we let $e^i$ be some  1-form over $TP$ such that $e^i(u)=\lambda^i(\pi_*(u))$. We can then find a basis of vertical one-forms by $\tilde e^\beta\circ\omega=e^\beta$ and obtain a 1-form in $P$. 
Allowing capital Roman indices to run over all types of indices, we then have the co-reference frame $\{e^A\}_{A=1}^n$ on $P$. We can use the metric $\bb G$ to dualize the frame and obtain $\{e_A\}_{A=1}^m$. Using the metric in this way, one can see that  $\pi_* (e_i)=X_i$ and, moreover, that $\{e^A\}_{A=1}^n$ becomes an orthogonal  co-frame that splits horizontal and vertical spaces orthogonally. That is, 
\be e^\sigma(e_i)=e^i(e_\sigma)=\lambda^i(\pi_*(e_\sigma))=0=\bb G(e_i,e_\sigma)\,.\ee
And since $\HH$ is a vector bundle orthogonal to $\V$, 
 $$\mbox{span}\left[\{e_\sigma\}_{\sigma=m+1}^n\right]=\V|_\theta\mbox{~~~e~~~span}\left[\{e_i\}_{i=1}^m\right]=\HH|_\theta\,.$$
The horizontal and vertical projections can then be written as: 
$$\hat H=e^i\otimes e_i, \qquad \hat{V}=e^\sigma\otimes e_\sigma\,.$$
It is important to note that these (co)frames need not be integrable; i.e., not tangent to any foliation of the bundle, and thus they only define a projection of tangent vectors and not of scalar functions on the bundle.

In the standard Kaluza--Klein construction, the derivations above illustrate that requiring a $G$-invariant inner product implies that the vertical directions are Killing once a connection $\omega$ and a group $G$ are specified. Later in Section~\S\ref{sec:KK_config}, this logic will be turned around: the group fibres of the kinematical metric of the theory will be assumed to be Killing. Defining these directions to be vertical and using orthogonality with respect to the kinematical metric to define horizontality will uniquely fix the connection-form $\omega$. 

\subsection{Kaluza--Klein curvature}\label{sec:KK_curv}

We now use the metric defined in the previous section to compute the Levi--Civita connection in $P$; i.e., the only torsionless one compatible with our metric. These computations require us to picture $P$ as a standard Riemannian manifold rather than a bundle.  The purpose of doing this will be to compare the geodesic structure of the base space $M$ with the projection of geodesic curves in the bundle $P$. Our first task will therefore be to compute the components of the bundle curvature defined by \eqref{eq:curv_h} so that we may compare this to the curvature of the Levi--Civita connection.

Let $\Omega\in\Gamma(\Lambda^2(TP^*)\otimes\fg)$ be the curvature-form in $P$, given in \eqref{eq:curv_h}, and the connection-form  $\omega$ be the vertical projection followed by the isomorphism between $\V$ and  $\fg$. 
As we have a local reference frame in $\V$, we can define the real-valued two-forms, $\Omega^\sigma\in\Gamma(\Lambda^2(\HH^*))$, as:  
$$\Omega^\sigma\otimes \tilde e_\sigma:=\Omega\,.$$ 
 We can further dismember this relation to obtain the real-valued components of the two-form. Since the form components of $\Omega$ are horizontal, as per \eqref{eq:curv_h2}, we define $F_{ij}^\sigma$ as
\be\label{fiji} (\frac{1}{2}F_{ij}^\sigma e^j\wedge e^i)\otimes \tilde e_\sigma:=(\Omega^\sigma)\otimes \tilde e_\sigma\,.\ee

We now turn attention to the Riemannian structure of $P$ and $M$ because we want to relate these to the bundle structure. Given the reference frame $e_A$, the easiest manner to capture the Levi--Civita connection is through the spin connection. That is: 
$$\d e^A=\w^A_B\wedge e^B\,. $$
The torsionless condition implies $\w_B^A=-\w_A^B$. Writing $\omega=e^\sigma\otimes \tilde e_\sigma$, we obtain 
\be\label{eq:important} \d\omega=\d(e^\sigma)\otimes \tilde e_\sigma=\w^\sigma_B\wedge e^B\otimes \tilde e_\sigma.\ee
Note that the derivative only acts on $e^\sigma$ because $\tilde e_\sigma$ is here a fixed element of the Lie-algebra. 

Using (\ref{eq:curv_h}), which ensures the curvature is purely horizontal as a differential form, we have: 
\be\label{a1}\Omega=\w^\sigma_i(e_j)e^j\wedge e^i\otimes \tilde e_\sigma\,.\ee
Comparing (\ref{a1}) and (\ref{fiji}) we obtain a first relation between the Riemman curvature of the bundle and gauge curvature of the bundle:
\be\label{a4} \w^\sigma_i(e_j)e^j\wedge e^i=  \frac{1}{2}F^\sigma_{ij}e^j\wedge e^i\,,\ee
and thus: 
\be\label{a6} \w^\sigma_i(e_j)=\frac{1}{2}F_{ij}^\sigma\,.  \ee 

Let ${^{\text{\tiny M}}\Gamma}$ be the Levi-Civita connection-form in $M$ relative to  $X_i=\pi_* (e_i)$ and the metric $h$. We then have:
\be \d\lambda^i={^{\text{\tiny M}}\Gamma}^i_j\wedge\lambda^j\,. \ee
Applying the pull-back $\pi^*$ (which commutes with $\d$, acting on the appropriate spaces) on both sides, we get:
\be\label{a7}\begin{array}{ll}
\d(e^i)& =\d(\pi^*\lambda^i)=\pi^*\d(\lambda^i)=\pi^*({^{\text{\tiny M}}\Gamma}^i_j\wedge\lambda^j)\\
~& =\bar\Gamma^i_j\wedge e^j\\
\d(e^i)& =\w^i_B\wedge e^B=\w^i_j\wedge e^j+\w^i_\sigma\wedge e^\sigma\,,
\end{array}\ee
where $\bar\Gamma^i_j:=\pi^*{^{\text{\tiny M}}\Gamma}^i_j$. Because $\bar\Gamma^i_j(u)={^{\text{\tiny M}}\Gamma}^i_j(\pi_*(u))$, it contains the representation of the base curvature only.  To fully display all the relations, we can  write down the  Christoffel symbols as $$\w^A_{BC}:=\w^A_B(e_C)\Longrightarrow \w^A_B=\w^A_{BC}e^C\,.$$
Now we apply \eqref{a7} to $(e_\alpha, e_k)$, and since $\pi_*(e_\alpha)=0$, we obtain $\w^i_j(e_\sigma)=\w^i_\sigma(e_j)$. Applying \eqref{a7} to $(e_i, e_j)$ we obtain:  $\bar\Gamma^i_j(e_k)=\w^i_j(e_k)$; and to $(e_\alpha, e_\beta)$ we get: $\w^i_\sigma(e_\beta)=0$ (and therefore $\w_i^\sigma(e_\beta)=0$). 

To obtain the full set of  relations between the gauge and the Riemann curvatures, we must collect more results. Using \eqref{eq:curv_h1}; i.e.,   $\d\omega=\Omega-\omega\wedge\omega$, and \eqref{eq:important}, we get:
\be\label{a5}\begin{array}{ll}
(\w^\sigma_B\wedge e^B)\otimes \tilde e_\sigma & =((\frac{1}{2}F_{ij}^\sigma e^j+\w^\sigma_{i}(e_\beta)e^\beta)\wedge e^i)\otimes \tilde e_\sigma-\frac{1}{2}e^\beta\wedge e^\nu[\tilde e_\beta,\tilde e_\nu]\\\\
\therefore~~\w^\sigma_B\wedge e^B & =\frac{1}{2}F_{ij}^\sigma e^j\wedge e^i+\frac{1}{2}C^\sigma_{\beta\nu} e^\beta  \wedge e^\nu\\\\
~& =\w^\sigma_\beta\wedge e^\beta + \w^\sigma_i\wedge e^i\,,
\end{array}\ee
and therefore, by (\ref{a4}) we have 
\begin{align}
\w^\sigma_\beta&=  -\frac{1}{2}C^\sigma_{\nu\beta}e^\nu\\
\w^\sigma_i&= -\frac{1}{2}F_{ij}^\sigma e^j\,. \label{eq:gather}
\end{align}
Finally, replacing \eqref{eq:gather} in (\ref{a7}) we get:
\be \w^i_j\wedge e^j=\bar\Gamma^i_j\wedge e^j -\frac{1}{2} F^\sigma_{ij}e^\sigma\wedge e^j\,, \ee
and therefore (since $\bar\Gamma^i_j(e_\sigma)=0$): 
$$\w^i_j(e_k)=\bar\Gamma^i_j(e_k), \quad\text{and}\quad \w^i_j(e_\sigma)= -\frac{1}{2} F^\sigma_{ij}\,.$$

Collecting all the results, we have:
\be\label{eq:main}
\left\{\begin{array}{ll}
\w^i_j= & \bar\Gamma^i_j-\frac{1}{2}\underset{\sigma}{\sum} F^\sigma_{ij}e^\sigma\\\\ 
\w^\sigma_i= & \frac{1}{2}F^\sigma_{ij}e^j\\\\
\w^\sigma_\beta= & -\frac{1}{2}C^\sigma_{\nu\beta}e^\nu
\end{array}\right.\ee
The decomposition above is central to the Kaluza--Klein analysis and underlies the decomposition that we make use of in our own analysis. The first of these equations is the most important. $\w^i_{j\sigma}\neq 0$ tells us that a horizontal vector, parallel transported along a vertical direction, will have its horizontal component rotated. Moreover, this rotation is precisely given by the gauge curvature.

The second equation tells us that a vertical direction does not rotate into a horizontal direction when transported along another vertical direction, but it does when transported along a horizontal direction. The amount of rotation is again given by the curvature.

The third equation tells us that the purely horizontal part of this transport reproduces the base-space transport. That is, the way a horizontal vector rotates in the horizontal direction, when transported along a horizontal direction, is identical to the analogous rotation on the base space. 

The decomposition \eqref{eq:main} is the main result necessary to express the Lorentz force in geometric terms. We will see in the next section that, for the trajectories of charged particles, the Lorentz force is manifested by the difference between the projection of the geodesics in $P$ and the geodesics in $M$. 

\subsection{The generalized Lorentz force}     
\label{sub:gen Lorentz}

In the relativistic case, the Lorentz force acquires a temporal component on top of the standard spatial one. The resulting force is simply proportional to the action of the curvature field $F$ over a given particle; i.e., $qF (v)$ for the (dual)-force exerted on the particle with 4-velocity $v$ and charge $q$. Apart from the standard spatial component, $qF (v)$ also contains a temporal term $q E_i v^i$, which is the work done by the field on the particle.

Now, let $\gamma:I\ra P$ be a geodesic and  $\bar\gamma=\pi(\gamma)$ be its projection to $M$. 
To compute the generalized Lorentz force, we first write the tangent to the geodesic curve as: 
\be\label{eq:geod_dec} \gamma'=u^i e_i +q^\sigma e_\sigma\,, 
\ee
where  $\pi_*(\gamma')=\bar\gamma'=u^i X_i$. Suggestively, we will call the function  $q:I\ra \fg$ given by  
\be q(t)=\omega(\gamma'(t))=q^\sigma(t)\tilde e_\sigma \label{eq:q_vert}\ee the curve's \emph{specific charge}.  Note that in the space-time case, we will use a prime, as in $\gamma'$,  to denote the derivative along the curve.  

In the general case, we will expect the Lorentz force to be of the form: 
\be\label{eq:Lor} F_j:= F_{ij} u^iq_\sigma\,.
\ee
To find the corresponding Lorentz-force deviation in the direction of the projected geodesic, we compare the projection of the acceleration along geodesics in $P$ with the acceleration along intrinsic geodesics in $M$: 
\be\label{eq:compare}\pi_*(\frac{\D^P\gamma'}{dt})= \pi_*({\nabla^P}_{\gamma'}\gamma') \qquad\text{and}\qquad \frac{\D^M\gamma'}{dt}= {\nabla^M}_{\bar\gamma'}\bar\gamma'\,,
\ee 
 and then take the difference.

We begin by proving a useful theorem.
\begin{theo}\label{ftnt:Killing}
  Let $X$ be a  Killing field in a Riemannian manifold $N$ and $\gamma$ a geodesic in $N$, then the inner product between $X$ and the tangent $\gamma'(t)$ is time-independent.
\end{theo}
\begin{proof}
$$\frac{d}{dt}\langle X, \gamma'(t)\rangle=\langle \nabla_{\gamma'}X, \gamma'(t)\rangle=-\langle \nabla_{\gamma'}X, \gamma'(t)\rangle=0$$
where we used the geodesic equation in the first equality and the Killing condition in the second (i.e. $\langle \nabla_{Z}X, Y\rangle=-\langle \nabla_Y X, Z\rangle$). 
\end{proof}
As a result of Theorem~\ref{ftnt:Killing}, a geodesic will maintain a constant angle with respect to a smooth Killing vector field. We therefore have:
$${\nabla^P}_{\gamma'}q^\sigma={\nabla^P}_{\gamma'}\bb G(e^\sigma, \gamma')=0\,.$$ 
Thus, 
\be \frac{D}{dt}\gamma':= {\nabla^P}_{\gamma'}\gamma'={\nabla^P}_{\gamma'}(u^ie_i)+q^\sigma {\nabla^P}_{\gamma'} e_\sigma\,,
\ee
and therefore:
\begin{eqnarray}
{\nabla^P}_{\gamma'}\gamma'&=&{u'}^i e_i+ u^i(u^j {\nabla^P}_{e_j}+q^\beta {\nabla^P}_{e_\beta})(e_i)+q^\sigma( u^j {\nabla^P}_{e_j}+q^\beta{\nabla^P}_{e_\beta})e_\sigma\nonumber\\
&=&{u'}^i e_i + u^i (u^j \w ^A_i(e_j)+q^\beta\w^A_i(e_\beta))e_A+q^\sigma (u^j\w^A_\sigma(e_j)+ q^\beta\w^A_\sigma(e_\beta)) e_A\,.\nonumber\\
\label{eq:geodesic_KK}\end{eqnarray}
We project this onto the horizontal components:
\begin{eqnarray}
e^k({\nabla^P}_{\gamma'}\gamma')&=&{u'}^k+u^i (u^j \w ^k_i(e_j)+q^\beta\w^k_i(e_\beta))+q^\sigma (u^j\w^k_\sigma(e_j)+ q^\beta\w^k_\sigma(e_\beta))\nonumber\\
&=& {u'}^k+u^i u^j \bar\Gamma^k_i(e_j)-u^iq^\beta F^k_{i\beta}\,,\label{eq:Lor_dev}
\end{eqnarray}
where we used \eqref{eq:main} in the last line. We are interested in $ \pi_*({\nabla^P}_{\gamma'}\gamma') $, as per \eqref{eq:compare}.  Moreover, the projection will only have components along $X_i$, which we can extract by contracting with $\lambda^i$. That is, we use $e^i=\pi^*\lambda^i$, and therefore $e^i({\nabla^P}_{\gamma'}\gamma')= \lambda^i(\pi_*{\nabla^P}_{\gamma'}\gamma'))$, to rewrite \eqref{eq:Lor_dev} as: 
\be \pi_*({\nabla^P}_{\gamma'}\gamma') =({u'}^k+u^i u^j  {^{~M}}\Gamma^k_i(X_j)-u^iq^\beta F^k_{i\beta}) X_k\,.
\ee
But we also have, 
$$\frac{\D^M\gamma'}{dt}= {\nabla^M}_{\bar\gamma'}\bar\gamma' = ({u'}^k+u^i u^j {^{~M}}\Gamma^k_i(X_j)) X_k\,,$$
and therefore \emph{the deviation between the geodesics is precisely given by the generalized Lorentz force}:
\be\label{eq:KK_Lorentz}
\pi_*({\nabla^P}_{\gamma'}\gamma')-{\nabla^M}_{\bar\gamma'}\bar\gamma'=-u^iq^\beta F^k_{i\beta} X_k\,
\ee
 as advertised. Figure~\ref{fig:deviation} gives an illustration of this effect.
\begin{figure}[h]
  \centering
  \includegraphics[width= 0.5\textwidth]{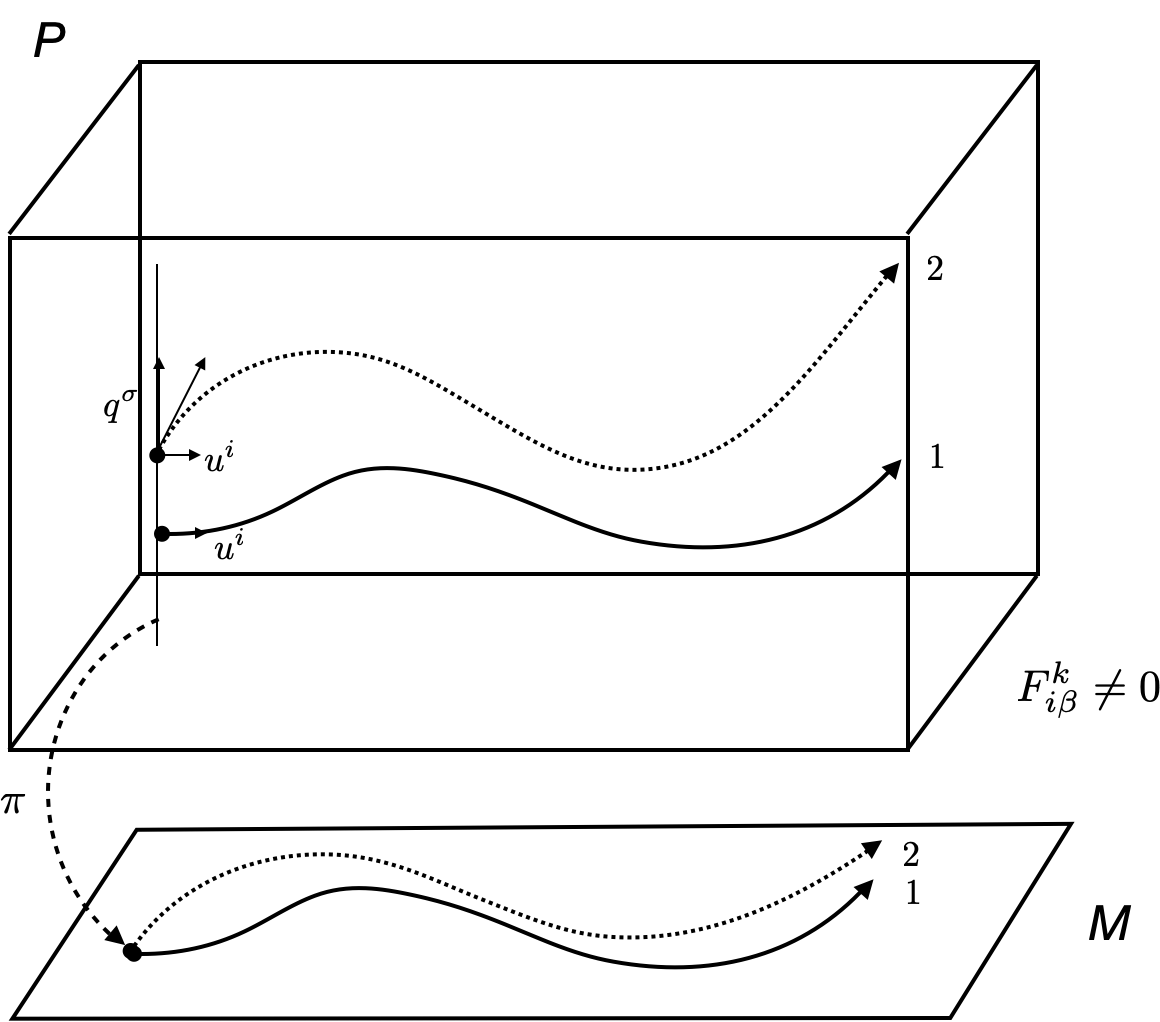}
  \caption{
    The difference between the projection of a curve with vanishing and non-zero vertical velocity $q^\sigma$ can be accounted for by a Lorentz force effect sourced by the bundle curvature according to \eqref{eq:KK_Lorentz}.
  }\label{fig:deviation}
\end{figure}

Although the original intent of Kaluza--Klein theory was to geometrize electromagnetism, it also serves to illustrate a different point: one can find a representation of the dynamics of the charged particle in a background field on a higher-dimensional space with no charge and no electromagnetic forces. However, the geometrical structures of the bundle are not faithfully represented by the intrinsic geometrical structures of the base space, and one must add information. This information is contained in the connection-form, through the curvature, and the vertical velocity, through the charge. 

The emergence of the Lorentz force in the manner described above is a consequence of the difference between a Kaluza--Klein formalism and the standard principal fibre bundle formalism, with fibers interpreted as internal spaces. In the Kaluza--Klein picture, the particles are embedded in  the extra dimensions; but, apart from being topologically and geometrically constrained, the extra dimensions are no different from macroscopic space-time dimensions. In the Kaluza--Klein formalism, the horizontal lift of the curve on the base space to the bundle does \textit{not} necessarily correspond to the motion that is projected.   More concretely, consider a geodesic $\gamma$ in the fibre bundle and its projection $\bar \gamma$ onto the base space. In general, the \emph{unique} horizontal lift, $\gamma_h$, of the projected curve $\bar \gamma$ will not be equal to $\gamma$ (and therefore not be geodesic). The only situation in which equality \emph{does} hold; i.e., when $\gamma_h = \gamma$, is when the charge $q$ is zero.

\section{Kaluza--Klein for configuration space}\label{sec:KK_config}

In the previous section, we reviewed the standard Kaluza--Klein construction that treats a higher-dimensional space-time as a principal fibre bundle. In that picture, the fibres represented extra dimensions and the base space was $3+1$ dimensional space-time. The logic assumed that the bundle connection was fixed because the extra dimensions, while required to be homogeneous and epistemically inaccessible, are otherwise assumed to be no different from the other known spatial dimensions. The goal in that case was to derive the Lorentz-force law on the base space from knowledge of the connection-form and the symmetries of the extra dimensions.

In this section we will flip this logic around. We will take the principal fibre bundle to be the Newtonian configuration space and the base space to be relative configuration space. In this construction, it is the kinematical metric rather than the bundle connection-form that is assumed to be given. The kinematical metric is assumed to have Killing directions along the symmetry group because that is precisely what it means for that group to be a time-independent symmetry of the Lagrangian. These Killing directions are then used to fix the vertical directions on the bundle. The horizontal directions are specified by orthogonality to the Killing directions according to the kinematical metric. This uniquely fixes the bundle connection. The connection defined in this way has been studied in the context of field-theories, where it was called the \emph{relational connection-form}.\footnote{The construction can in principle be applied to any (even infinite-dimensional) configuration space with a Lie group symmetry action that preserves the kinetic term. This was fully described for Yang-Mills theories,  also within the symplectic framework,  in a series of papers \cite{GomesRiello2016, GomesRiello2018, GomesHopfRiello, GomesRiello_new, GomesStudies}, where it was dubbed the \emph{Singer--deWitt connection-form}. \label{ftnt:varpi}} 

The goal now will be to organize the terms in the force law on the base space in a geometrical way; i.e., as arising partly from the dynamics intrinsic to base space and partly as descending from dynamics of the bundle. What we find in this new set-up is that, in addition to the expected Lorentz-force term, there is a new term due to the base-space dependence of the bundle metric along the fibres. This term generalises the Kaluza--Klein construction, where the bundle connection is fixed in a way that respects more closely the properties of the bundle metric. Our construction, while inspired by the Kaluza--Klein formalism, is therefore distinct from it and more general.

In Section~\S\ref{sec:config_varpi}, we will start with an overview of the principal fibre bundle construction in the context of general configuration spaces. We invert the results of Section~\S\ref{eq:bundle_metric} and derive a connection-form,  equation~\eqref{eq:varpi_abstract_solution}, from the kinematical metric. Finally, in Section~\S\ref{sec:curvature}, we will derive the induced bundle curvature, \eqref{eq:curv_formula}.

\subsection{The connection-form defined by orthogonality}\label{sec:config_varpi}

The purpose of this section will be to demonstrate that for any configuration that has a metric that is invariant under the action of some group, we can induce a connection by orthogonality with respect to that metric and the vertical spaces, which are canonical in the sense that they depend only on the group action. In the more general setting introduced in this section, configuration space can even be a field-space. For example, a natural gauge-invariant metric exists in the configuration space of Yang-Mills theory and that induces a connection (cf. \citep[Sec. 4]{GomesHopfRiello}, and references therein). To accommodate this amount of generality and to highlight the differences between the configuration-space approach used here and the space-time-based approach of Kaluza--Klein theory, we will now introduce some notation to distinguish the basic mathematical structures.

The space-time-based fibre bundle, which we previously called $P$, will now be given by a configuration space that we denote by $\F$. To denote those  quantities that now belong to configuration space and not to space-time we will employ the double-struck notation of \citep{GomesHopfRiello}. For instance, $\bb X \in \mathfrak{X}^1(\F)$ is a once continuously (functionally) differentiable vector field \textit{in configuration space}, and $\fLie_{\bb X}$ is a Lie-derivative along $\bb X$ \textit{in configuration space}, $\dd F$ is an exterior functional derivative \textit{in configuration space},  and so on. To further distinguish the configuration space case from the space-time one, we denote the connection-form not by $\omega$ (given in \eqref{eq:connection}), but by $\varpi$, and the curvature 2-form not by $\Omega$ (given in \eqref{eq:curv_h} or \eqref{eq:curv_h2}), but by $\bb F$.  Lastly, in the configuration space context, since we will be using it more often, it is useful to replace the $\iota$ map, inducing a bijection between the Lie-algebra and the vertical spaces ($\iota$ is defined  by the action of the group on the configuration space, as in \eqref{eq:iota}), by a more direct notation: $\iota( \xi)\equiv\xi^\#$. 

A configuration-space metric and a vertical direction supply enough ingredients to define a connection if and only if the directions along the orbit are Killing; i.e., the fundamental vector fields $\xi^\#$ must be Killing. Given a group of transformations, $\cal G$  (possibly infinite-dimensional), 
    \begin{align}
    \fLie_{\xi^\#} \bb G = 0
    \qquad \text{for all} \;
    \xi \in \text{Lie}(\G)\,.
    \label{eq_GG_Killing}
    \end{align}

For a field-space metric $\bb G$, we define the associated connection $\varpi$ by demanding the following orthogonality relation:\footnote{The following equation holds pointwise on $\mathcal{Q}$, where the vector field $\bb X$ identifies, for each configuration $\varphi\in \mathcal{Q}$, a tangent vector $\bb X_\varphi \in {\rm T}_\varphi \F$. In the main text we have omitted the subscripts.}  
  \begin{equation}
  \bb G(\xi^\#, \hat H(\bb X)) \equiv \bb G(\xi^\#, \bb X - \varpi(\bb X)^\#) = 0,
  \label{eq4.3}
  \end{equation}
for all $\xi \in \text{Lie}(\G)$ and all $\bb X \in \mathfrak{X}^1(\F)$. As before, $\hat H$ stands for the horizontal projection induced by $\bb G$ and, writing $\varpi(\cdot)^\#=\hat V(\cdot)$, the equation becomes trivial. 
  
Formally,  equation \eqref{eq4.3} can be solved for $\varpi$ as follows. Let $\bb Q_{ab}$ be the pullback to $\fG$ under $\cdot^\#$ of the metric induced from $\bb G$ on the fibres as expressed in the $\{\tau_a\}$ basis of the Lie-algebra:
  \be
  \bb Q_{ab}=\bb G(\tau_a^\#, \tau_b^\#),
  \label{eq_Qab}
  \ee
and $\bb Q^{ab}$ its inverse. Note that $\bb Q_{ab}$ does not in general coincide with the (point-wise extensions of the) Killing form in $\fg$ (the natural inner product of the Lie-algebra).

  Expanding $\varpi = \varpi^a \tau_a$, equation \ref{eq4.3} can be written as $\bb G(\tau_a^\#, \bb X) = \bb Q_{ab} \fI_{\bb X} \varpi^b$, which is readily inverted as
  \begin{align}\label{eq:varpi_abstract_solution}
  \varpi = \bb Q^{ab} \bb G (\tau_b^\#, \cdot) \tau_a\,.
  \end{align}
  This is \emph{the main result of this section}. Note that $\bb G(\xi^\#, \cdot)$ accepts field-space vectors and hence defines a one-form in field-space.

From the last equation, we immediately obtain the first fundamental property required of a connection-form: $ \varpi(\xi^\#) = \xi$. See \cite[Sec. 4.1, equations 4.7 and 4.8]{GomesHopfRiello} for a proof that an $\varpi$ defined by the procedure above will also transform correctly under gauge transformations if $\xi^\#$ is a Killing vector of $\bb G$. As we saw in \eqref{eq:hor_proj_v}, these properties combine to ensure that $\hat {H}(\bb X)$ has the correct ``vertical corrections'' that render it invariant also under time-dependent (or, more generally, field-dependent) gauge transformations. They also ensure that for any one-form in configuration space, $\bb \lambda$,  the covariant exterior differential $\dd_\varpi\bb \lambda=\dd\bb\lambda-\varpi\wedge\lambda$ is  covariant.   
      
To summarize, a field-space metric determines a vertical projector by providing a notion of orthogonality. If gauge transformations preserve orthogonality to the fibres, then the vertical projector gives a connection.

  \subsection{The curvature of the connection-form\label{sec:curvature}}

  A natural question to ask is how the properties of a field-space connection are linked to the properties of the field-space metric that determines it. In particular, one may ask if the curvature of the field-space connection can be calculated directly from the field-space metric in a useful way. The answer is affirmative, as was shown in \cite{GomesHopfRiello}.
  
  The intuition is the following: $\varpi$ contains information about the horizontal planes, which are the planes orthogonal to the gauge orbits. If those planes can be integrated in the sense of Frobenius' theorem to (infinite-dimensional) hyper-surfaces, then $\varpi$ is flat. The curvature $\fF$ of $\varpi$ corresponds to the anholonomicity, or non-integrability, of the planes orthogonal to the gauge orbits. Thus it is possible to obtain the curvature directly from the metric.
  
  The resulting relationship between a field-space-metric $\bb G$ and the curvature $\fF$ of the associated $\varpi$ is: 
  \begin{equation}\label{eq:curvature_and_metric}
  \bb G \big( \bb F(\bb X, \bb Y)^\#, \xi^\#\big) = \dd (\bb G (\xi^\#)) (\hat H (\bb Y), \hat H (\bb X))
  \quad \text{for all} \;
  \xi \in \text{Lie}(\G), ~\dd \xi = 0,
  \end{equation}
  and any $\bb X, \bb Y \in \mathfrak{X}^1(\Phi)$. On the right hand side, $\bb G(\xi^\#)\equiv\bb G(\xi^\#,\cdot)$ is a one-form on field-space, so $\dd \bb G (\xi^\#)$ is a two-form. By horizontally projecting the dummy vector fields $\bb X, \bb Y$ on the right hand side, we are taking the horizontal-horizontal part of that two-form. Formally solving for $\bb F$, we get
  \begin{align}
  \bb F = \bb Q^{ab} \big(\dd \bb G(\tau_b^\#) \big)_{HH} \tau_a\,,\label{eq:curv_formula}
  \end{align}
  which is \emph{the main result of this section.} Note that, in these formulas, $\dd$ acts on the one-form $\bb G(\xi^\#)$. Even if $\xi$ is taken to be configuration-independent; i.e., $\dd \xi=0$, the operator $\cdot^\#$ generically introduces configuration-dependence. As we will see, this a feature of the hash operator for rotations --- but not for translations. For a proof of \eqref{eq:curvature_and_metric} (and thus the origin of \eqref{eq:curv_formula}) see \cite[Sec. 4.2, equation 4.12]{GomesHopfRiello}.

 We have now developed the technical machinery required to pinpoint the kinds of modifications to the force law on the base space that can be expected from the projection of the dynamics on the bundle. Firstly the Lie algebra can act in a field-space dependent way. This is not the case in Abelian gauge theories, but it is the case for non-Abelian ones.\footnote{In Yang--Mills, $\delta_\xi A= \d \xi + [A,\xi]$ involves the gauge potential $A$, while in general relativity,  $\delta_\xi g = \pounds_\xi g$  involves the metric $g$ (where now $\xi\in\mathfrak{X}^1(M)$ and $\pounds_\xi$ is the space-time Lie derivative along $\xi$). } For rotations, $\delta_\xi \mathbf r= \xi\mathbf{r}_\alpha$, where $\mathbf{r}_\alpha\in \bb R^{3}$ with $\alpha=\{1, \cdots, N\}$ and $\xi\in \mathfrak{so}(3)$. This dependence will in general produce curvature of the connection constructed by orthogonality with respect to the gauge orbits, that is, it implies $\bb F\neq 0$. Secondly, the inner product along the orbits may be base-space dependent, which  will create a potential force for the motion on base space. This is neither the case for translations nor for Kaluza--Klein in space-time, where the vertical kinetic energy is rigidly fixed by the Lie algebra. In those cases, the orbits are metrically the same everywhere. These two effects are important because they are directly linked to the two terms of the main result of this paper: the modified force law, \eqref{config_Lorentz}, derived for rotations in the next section.

\section{ Charging the relationalist for angular momentum}
\label{sec:LR}

 In this section, we take the general formalism of the previous section and apply it in the context of translations and rotations in Newtonian mechanics. The goal is to project the dynamics of the system onto relative configuration space. We obtain a generalised deviation equation analogous to the Lorentz force law of \eqref{eq:KK_Lorentz}. In the case of translations, the group action is field-independent, and so are the orbit metrics, and thus the bundle curvature is exactly zero and there are no potential terms on base space, and therefore the bundle projection is trivial. In the case of rotation, the bundle curvature is non-zero and the orbit metric depends on the point on base space, thus we find two independent terms that couple to a conserved $SO(3)$-charge. This charge is the cost of expressing effects due to the total angular momentum in a Newtonian system directly in relational terms.

In our construction, we take the principal fibre bundle to be the configuration space $\F\simeq \bb R^{3N}$ of $N$ point particle positions in 3 dimensions. In line with the assumption of Section~\S\ref{sec:KK_config}, this space is endowed with a canonical inner product that is invariant under the action of both translations and rotations. This canonical inner product is precisely the kinematical inner product used to construct the Newtonian kinetic energy, and the potential energy function is usually assumed to reflect these kinematical symmetries. 

We can therefore take the structure group for $\F$ to consist of either translations and/or rotations, and in each case we can find a connection-form through the introduction of an equivariant horizontal sub-bundle of $T\F$, as per Section~\S\ref{sec:config_varpi}. Given these choices, the group orbits define vertical spaces as the tangent to these orbits. The fixed kinematical metric of Newtonian mechanics then defines the orthogonal spaces to the orbits as the horizontal complement of $V\subset T\F$. This then leads to a fixed connection-form and curvature, which we compute in Section~\S\ref{sec:LR_intro} following the procedure of Section~\S\ref{sec:config_varpi} and Section~\S\ref{sec:curvature} . The bundle curvature we obtain enters the desired deviation relation: equation~\ref{config_Lorentz}.

 \subsection{The rotation bundle}\label{sec:LR_intro}
 
 While many elements of this construction have been beautifully laid out in \cite{LittlejohnReinsch}, one of the great advantages of the present construction is that it does away with `rotational frames' for the motion. We have now developed enough mathematical machinery to explain some of these advantages in more detail.
   
The construction of \citet{LittlejohnReinsch} employs specific frames with which to measure the rotation. Transformations of these frames are seen as `passive'; i.e., they are gauge transformations taking one section of the bundle to another. In the principal fibre bundle formalism, we do away with the definition of frames and concentrate on the active definition of symmetry transformations. This immensely simplifies the treatment while still capturing the relevant effects in the base space.

In the principal fibre bundle picture, a choice of frame of the kind used in the bulk of \citet{LittlejohnReinsch} provides a section of the bundle $s:\F/\G\rightarrow \F$.\footnote{When pulling back the constructions of Section~\S\ref{sec:PFB_intro} on the principal fibre bundle to this section, one deals with  objects on the associated bundle such as the connection $A=s^*\omega$ and curvature $B=s^*\Omega$.} Explicitly constructing such a section is non-trivial work, which we will not require in this paper. Without a section, we are only able to talk about a total change in orientation of a system for closed curves in configuration space. For an open curve, one must establish an orientation convention to describe the overall change between the start- and end-points of the curve. This structure is enough to describe the infinitesimal change of orientation (i.e., the rotation) along a curve. It is also sufficient to describe curvature and the covariant divergence of the potential. 
 
When we let the group of rotations act on a given configuration, we assume an active view: we are not changing the orientation convention for the relation between the inertial space frame and the body frame. Instead, we are  rotating the configuration. Thus, if a Newtonian configuration represents an inertial frame, the rotated configuration will represent a rotated inertial frame.\footnote{ Although the difference may seem inconsequential, it in fact implies a \textit{right} action of $R\in \SO(3)$ on the $\mathbf {r}_\alpha$, not a left one (cf. \citep[p. 244]{LittlejohnReinsch}).}
 
 We start with $N$ particles on $\bb R^3$, with coordinates  in a given inertial frame $\rr_\alpha$, $\alpha=1, \cdots N$, and masses $m_\alpha$. The Lagrangian will be of the form 
 \be\label{eq:Lagrange}\mathcal{L}=K(\rr, \dot\rr)-V(\rr).
 \ee
  The total kinetic energy of the system, which will determine our configuration space metric, is 
 \be\label{eq:kinetic}K(\dot\rr)=\frac12\sum_\alpha m_\alpha |\dot\rr_\alpha|^2=\frac12\sum_\alpha m_\alpha (\dot\rr_\alpha\cdot \dot \rr_\alpha)\,, \ee
 where $\cdot$ is the Euclidean inner product in $\R^3$. The kinetic term $K$ is equivalent to a choice of inner product:
 \be\label{eq:config_ip}\bb G(\dot\rr, \dot\rr')= \frac12\sum_\alpha m_\alpha \dot\rr_\alpha\cdot \dot \rr'_\alpha.\ee
 Note that, in the relativistic Kaluza--Klein case, we used primes to denote derivatives along the curve. Here, in the non-relativistic configuration space of $N$ particles, we will revert to the standard dot notation; i.e., $\dot \rr_\alpha$. 
 
 Equation \eqref{eq:config_ip} is clearly invariant under the time-independent transformations: 
 \begin{subequations}
 \begin{align}
  \rr_\alpha&\mapsto \rr_\alpha+\mathbf{v}\\
   \rr_\alpha&\mapsto \rr_\alpha R \,.
\end{align} \end{subequations}
Infinitesimally, the group actions of $\bb R^3$ and $\SO(3)$ correspond to, respectively (for $\mathbf v\in \R^3$ and $\xi\in \mathfrak{so}(3)$):
\begin{subequations}\label{eq:gts}
 \begin{align}
  \delta_{\mathbf v}\rr_\alpha &=\mathbf{v}\\
   \delta_{\xi}\rr_\alpha &=\rr_\alpha \xi\,.
\end{align} \end{subequations}
In terms of the Lie derivative (where $\delta_\xi\rr=:\xi^\#$),
$$\bb L_{\xi^\#}\bb G(\bb v, \bb v' )=\frac{1}{t}\lim_{t\rightarrow 0}(\bb G(\bb v, \bb v')- \bb G(\bb v \xi, \bb v' \xi))=0\,.
$$
In other words: the metric has Killing directions along $\xi^\#$. 

 By constructing a connection-form in configuration space, we can extend the invariance \textit{of the horizontal part of the metric} to one under time-dependent transformations. A separate question exists about the physical meaning of such connection-forms. 
 
For now, it suffices to note that a connection-form as defined by orthogonality to the orbits, as in section \ref{sec:config_varpi}, \textit{is} dynamically preferred because the metric in configuration space defines the kinetic energy. To find what the horizontality condition implies, we consider purely vertical directions; i.e., tangent vectors along the orbit. Orthogonality with respect to the translation orbits gives:
\be\label{eq:ortho_trans}  \bb G(\dot \rr,  \delta_{\mathbf v})=  \frac12\sum_\alpha m_\alpha \dot\rr_\alpha\cdot  \mathbf{v}=0, \qquad \forall  \mathbf{v}\in \R^3
\ee
This implies horizontal velocities $\rr_\alpha=\cc_\alpha$ must obey the constraint:
 $$\sum_\alpha m_\alpha \dot\cc_\alpha=0\,;$$
i.e., the \emph{vanishing of the total linear momentum}. By fixing the integration constant to zero, we quotient out translations by going to a centre of mass frame, $\cc_\alpha$, which obeys:
  $$\sum_\alpha m_\alpha \cc_\alpha=0.$$
 
 A rotation $R\in \SO(3)$ acts on this translationally reduced system, but the horizontality condition is a little more complicated. If $Q(\varphi, \boldsymbol {n})$ denotes a counterclockwise rotation with angle $\varphi$ about the axis specified by the unit vector $\boldsymbol {n}$, then

$${\displaystyle \left.{\frac {\operatorname {d}} {\operatorname {d} \varphi} }\right|_{\varphi =0}Q(\varphi ,{\boldsymbol {n}}){\boldsymbol {x}}={\boldsymbol {n}}\times {\boldsymbol {x}}}$$
for every vector $\boldsymbol {x}\in \R^3$.
This can be used to show that the Lie algebra $\mathfrak{so}(3)$ (with commutator) is isomorphic to the Lie algebra $\mathbb {R} ^{3}$ (with cross product). If we take $J_a$ as a basis of the $\mathfrak{so}(3)$ Lie-algebra,\footnote{In more detail, a most often suitable basis for $\mathfrak{so}(3)$ as a 3-dimensional vector space is
$${\displaystyle J_{\mathbf {x} }={\begin{bmatrix}0&0&0\\0&0&-1\\0&1&0\end{bmatrix}},\quad J_{\mathbf {y} }={\begin{bmatrix}0&0&1\\0&0&0\\-1&0&0\end{bmatrix}},\quad J_{\mathbf {z} }={\begin{bmatrix}0&-1&0\\1&0&0\\0&0&0\end{bmatrix}}.}$$
The commutation relations of these basis elements are,$$
{\displaystyle [J_{\mathbf {x} },J_{\mathbf {y} }]=J_{\mathbf {z} },\quad [J_{\mathbf {z} },J_{\mathbf {x} }]=J_{\mathbf {y} },\quad [J_{\mathbf {y} },J_{\mathbf {z} }]=J_{\mathbf {x} }}$$
which agree with the relations of the three standard unit vectors of $\mathbb {R} ^{3}$ under the cross product. \label{ftnt:basis}} 
we have
\be\label{eq:Jhash} J_{(a)}^\#=\boldsymbol {e}_{a}\times \cc_\alpha\,,\ee
where $\boldsymbol {e}_{a}$ is the unit vector along the direction $a$. 

Thus, the analogous condition to \eqref{eq:ortho_trans} for rotations becomes:
\be\label{eq:ortho_rot}  \bb G(\dot \cc,  \delta_{\boldsymbol {n}}\cc)=  \frac12\sum_\alpha m_\alpha \dot\cc_\alpha\cdot ({\boldsymbol {n}}\times \cc_\alpha) =0, 
\ee
for $\boldsymbol {n}$ a unit vector. Equation \eqref{eq:ortho_rot}  gives the following condition for velocities orthogonal to the fiber:\footnote{This identity is trivial if written in components (i.e. using the totally-antisymmetric tensor in three-dimensions: $(\epsilon_{ijk} n^i c^j)\dot c^k=-(\epsilon_{ijk} \dot c^i c^j)n^k$.}
\be \mathbf{L}=\sum_\alpha m_\alpha \dot\cc_\alpha\times \cc_\alpha=0. 
\ee
where $ \mathbf{L}$ is of course the angular momentum vector. In short, a motion on the (translationally-reduced) configuration space $\R^{3N-3}$ is orthogonal to the rotation orbits \emph{if and only if the associated angular momentum of this motion vanishes}.

 \paragraph{The horizontal and vertical projections}
 
 In order to use our formalism to obtain the connection-form and curvature, \eqref{eq:varpi_abstract_solution} and \eqref{eq:curv_formula} respectively, we need the vertical metric $\bb Q^{ab}$, where $a, b$ are the indices parametrizing the basis of infinitesimal rotations ($a,b=x, y,z$). For translations, $\bb Q$ is trivial: $\bb G(\delta_{\mathbf v},  \delta_{\mathbf v}')=  \mathbf v\cdot \mathbf v'$, which is completely independent of the translationally-reduced configuration. 
 
 For rotations on the other hand, we obtain:
 \begin{align}  \bb G(  \delta_{\boldsymbol {n}}\cc,    \delta_{\boldsymbol {n}'}\cc) &=  \frac12\sum_\alpha m_\alpha ({\boldsymbol {n}}\times \cc_\alpha)\cdot ({\boldsymbol {n}'}\times \cc_\alpha)\\
 &=  \frac12\sum_\alpha m_\alpha\left(|\cc_\alpha|^2\delta_{ij}-c_{\alpha i} c_{\alpha j}\right)n^i {n'}^j\,.
 \end{align}
 This yields a vertical metric, which is just the moment of inertia tensor:
 \be\label{eq:MIT}
 \bb Q_{ab}=M_{ab}=\frac12\sum_\alpha m_\alpha\left(|\cc_\alpha|^2\delta_{ab}-c_{\alpha a} c_{\alpha b}\right)\,.
 \ee
Because 
  \be\label{eq:intermediary}\bb G(\cdot, J_b^\#)= \frac12\sum_\alpha m_\alpha \dd c^i_\alpha  c^k_\alpha \epsilon_{ibk}), 
\ee
where $\dd c^i_\alpha$ is a basic configuration space 1-form (like $\d x$ would be in space-time), we obtain the following expression for the connection-form from \eqref{eq:varpi_abstract_solution}:
 \be\label{eq:varpi_rot}
 \varpi(\dot \cc)=(M^{-1})^{ab} \left(\frac12\sum_\alpha m_\alpha \epsilon_{bij} c_\alpha^i \dot c_\alpha ^j\right ) J_{(a)}\,.
 \ee
 Given an infinitesimal change of configuration, \eqref{eq:varpi_rot} provides the necessary rotation for that change to carry no angular momentum. By the properties of the connection-form --- arising from orthogonality to the fibre with respect to a $\G$-invariant metric, discussed in section \ref{sec:config_varpi} --- this adjustment is gauge-covariant; i.e., it does not depend on the orientation of the configuration we started from. The connection-form defines a standard of orientation infinitesimally along a curve.
 
It is also easy to  write the vertical projection. From \eqref{eq:varpi_abstract_solution}, we have $\hat V=\bb Q^{ab} \bb G (J_b^\#, \cdot) J^\#_a$.\footnote{By writing the basis $J^\#_a=|\ell_a\rangle$ and $\bb G (J_b^\#, \cdot) =\langle \ell_b|$ to match \citet{LittlejohnReinsch}'s notation, we obtain their equation 5.42: $\Pi_V= |\ell_a\rangle (M^{-1})^{ab} \langle \ell_b|$.} If use equation~\eqref{eq:MIT} and insert \eqref{eq:Jhash} into \eqref{eq:varpi_rot}  we obtain: 
\be\label{eq:vert_proj} \hat V(\dot \cc)=\varpi(\dot\cc)^\#=\sum_\alpha m_\alpha \dot\cc_\alpha\times \cc_\alpha=\mathbf{L}\,,
\ee
which is the angular momentum. The horizontal projection is just its complement: $\hat H=\bb 1-\hat{V}$. Given a generic centre-of-mass configurational velocity $\dot \cc_\alpha$, the corrected velocity $\dot \cc_\alpha-\hat V(\dot \cc)$ has vanishing angular momentum. 

Curvature implies that, for a closed loop in the base space, the orientation may change even for motion with zero angular momentum. We can now write the curvature using \eqref{eq:curv_formula}. First, from \eqref{eq:intermediary}:
$$\dd\bb G(\cdot, J_b^\#)= \frac12\sum_\alpha m_\alpha \epsilon_{ibk}\,\dd c^i_\alpha\curlywedge \dd c^k_\alpha, 
$$ where $\curlywedge$ is the exterior differential for forms in configuration space. From this equation it is apparent that if the group action on configuration space did not depend on the configuration, the exterior derivative $\dd$ would have nothing to act non-trivially on. We would then obtain $\dd\bb G(\cdot, J_b^\#)=0$, and vanishing curvature as a consequence. This is what occurs for translations. As it stands, the curvature for the rotational bundle can be written as:
\be \bb F(\dot c, \dot c')=(M^{-1})^{ab} \left(\frac12\sum_\alpha m_\alpha \epsilon_{bij} \hat{H}(\dot \cc_\alpha)^i \hat H(\dot \cc'_\alpha )^j\right ) J_{(a)}\label{eq:curv_rot}\,,
\ee
which only depends on the base space through the moment of inertia tensor and the horizontal projections.\footnote{Two other facts that are not important for us here, but deserve mention: by construction,  $\bb F$ obeys the Bianchi identity, i.e. for the gauge-covariant exterior derivative, $\dd_\varpi \bb F:=\dd \bb F +[\bb F, \varpi]=0$. Moreover, the covariant base space divergence of $\bb F$, which would, in the space-time context, give rise to the sourced Yang-Mills equation, $\dd_\varpi^\dagger\bb F=\rho$, where $\dd_\varpi^\dagger$ is the adjoint of $\dd_\varpi$ under the base space projection of $\bb G$, here  can be shown to always vanish for the rotations (cf. \citep[Eq. 5.2]{LittlejohnReinsch}).}

\paragraph*{Translations}

For translations, the generators of the algebra $\bb R^3$ can be taken to be the unit vectors, $\tau_a=e_a$, given by $\{e_x, e_ y, e_ z\}.$ The vertical metric $\bb Q^{ab}=\delta^{ab}$, and therefore 
$$\varpi_{\text{\tiny trans}}(\dot \rr)= (\sum_\alpha m_\alpha \dot \rr_\alpha \cdot e_a) e_a= \sum_\alpha m_\alpha \dot \rr_\alpha, 
$$ since $\sum |e_a\rangle\langle e_a|$ is the identity operator. 

Analogously to the rotational case, the connection-form defines, infinitesimally along the trajectory of the system, the standard of linear translations. The translational connection therefore yields the linear momentum of the configurational velocity. Horizontal motion, in analogy to the rotational case, coincides with a choice of coordinate system for which the total \textit{linear} momentum vanishes. In other words, for an arbitrary velocity $\dot \rr_\alpha$, the corrected velocity $\dot\rr_\alpha-\varpi_{\text{\tiny trans}}(\dot \rr)$ has vanishing linear momentum. From \eqref{eq:curv_formula}, the curvature clearly vanishes, since neither $\bb G$ nor $\tau_a^\#$ depend on the configuration. This is the main difference compared with the rotational case.

\paragraph*{Relation to Barbour-Bertotti theory}

First, we establish the relation between $\varpi$ and best-matching. The $\bb G$-induced connection-form $\varpi$, here illustrated for point particles under rotations and translations (but also available more generally for field-theory), performs, through horizontal projection, precisely what is generally recognized in the philosophical literature as the job of \textit{best-matching} (cf. \cite{Barbour_Bertotti},  \citet[Ch. II.5]{Flavio_tutorial} and references therein, and \cite{gomes_riem} for the original expression of best-matching in terms of a connection-form). The idea there is the same one as here: given neighboring orbits, a best-matched infinitesimal trajectory between them is one that extremizes some norm on the difference. For the standard, equivariant version, the norm must be $\G$-invariant, and is in fact the one induced by $\bb G$. Given some representative of a velocity between neighboring fibres; e.g., $\dot \cc_\alpha$, the best-matched velocity is 
$$\hat H(\dot\cc)\equiv \cc_\alpha-\varpi(\dot\cc)^\#\,,$$
which, by \eqref{eq:hor_proj_v}, is gauge-covariant. By itself it is not gauge-invariant: it only renders a horizontal kinetic term that is already invariant under time-independent transformations completely  gauge-invariant, as we will see below.  

In Barbour--Bertotti theory, the motion of the universe must be describable by a best-matched trajectory in the above sense so that the total angular momentum of the universe is constrained to be zero. A non-zero angular momentum cannot be accommodated because the vertical metric depends on it, and therefore a non-trivial vertical velocity (i.e., an angular momentum) can influence the dynamics on the bundle. This constraint is problematic: isolated subsystems of the universe clearly can have non-zero angular momentum, and thus lie outside the scope of the theory. Barbour--Bertotti theory is, in this sense \textit{cosmological}. In particular, it does not satisfy subsystem-recursivity as advocated in \cite{pittphilsci16623} because subsystems can have angular momentum while the universe cannot. 
 
Here we have removed these limitations. Indeed, from the perspective of the bundle, the kinematic term in the action determines a standard of rotation \textit{at each infinitesimal step or change} along a curve. This standard is established for both zero and non-zero angular momentum.

From the above, let the set of particles indexed by $\alpha$ split into two sets: $\alpha_I\in \Lambda_I$ and $\alpha_{II}\in \Lambda_{II}$. Starting with the $\cc_\alpha$ coordinates of the entire system, we define the individual centres-of-mass of systems $I$ and $II$, $q_{I}$ and $q_{II}$, and new coordinates $\cc_\alpha^{I}$ and $\cc_\alpha^{II}$ such that: 
\be \cc_\alpha=\begin{cases} 
\cc_\alpha^I+q^I\quad\text{for}\quad \alpha\in \Lambda_I\\
\cc_\alpha^{II}+q^{II}\quad\text{for}\quad \alpha\in \Lambda_{II}\,,
\end{cases} \ee
where e.g. $q^{I}:=-\frac{1}{\sum_{\alpha\in \Lambda_{I}} m_\alpha}\sum_{\alpha\in \Lambda_{I} }m_\alpha \cc_\alpha$. This formalism applies if and only if there are no external torques acting on the system. This occurs if $\|q_I-q_{II}\|>>\|\cc_\alpha^{I, II}\|, \quad \forall \alpha$. More colloquially, this condition will hold if the distance between the clusters of particles is much greater than the distance between the particles of the same cluster.\footnote{Under these conditions, for radial potentials, it will generally be the case that $V\approx V_I(\cc_\alpha^I)+V_{II}(\cc_\alpha^{II})+V(q_I, q_{II})$ where $q_I$ and $q_{II}$ are the centres of mass of the two systems. }

\subsection{The projected dynamics}\label{sec:KK_proj}

In this section we will project the dynamics in the absolute Newtonian configuration space to relative configuration space. We will compare the projection of geodesic curves in the bundle with different vertical velocities to geodesic curves on the base space. The difference in the acceleration of these curves will give a deviation equation analogous to the Lorentz force seen in standard Kaluza--Klein theory (Equation~\ref{eq:KK_Lorentz}). But before doing this, let us consider some of the general properties of the dynamics under this decomposition.

Using the vertical and horizontal projection of the motion, and assuming the potential term depends only on the relative configuration space, we can rewrite the Lagrangian from  \eqref{eq:Lagrange} as: 
\be \mathcal{L}=K(\hat{H}(\dot \cc))+K(\hat{V}(\dot \cc))- V([\cc])=K_V+K_H- V\,,
\ee
where the potential term should only depend on the orbit and not on the particular configuration (i.e., the potential does not depend on the orientation or absolute position of the configuration).  

Note that, by the properties of $\varpi$, $K_H$ is fully covariant --- even under transformations of the type $g_t:=g_0\exp(t\xi)$ (cf. \eqref{eq:hor_proj_v}). Since the kinetic term involves only a $\G$-invariant metric, $K_H$ is fully gauge-invariant. The potential term is invariant under both time-dependent and time-independent symmetry transformations. 

The vertical term could be problematic. But as we saw in Theorem~\ref{ftnt:Killing}, the vertical velocity is conserved along geodesic motion (and it is a small step to show that it is also invariant if one adds a gauge-invariant potential). Thus, even if $K_V$ is not gauge-invariant; i.e., not purely determined by the relations, it will only require dim($\mathfrak{g}$) extra constants of motion along the base space trajectories. Each component of the angular and linear momenta will therefore be conserved. Nonetheless, going beyond a mere accounting of degrees of freedom, the vertical kinetic term can have a non-trivial influence on the projected dynamics, and this is the second sort of effect we need to account for. 

Usually the vertical inner product of the configurational metric $\bb Q$ is quite different from the one present in the standard space-time Kaluza--Klein theory,  which is induced by the Killing form in \eqref{kil}. Likewise, the horizontal metric is the  space-time metric in standard Kaluza--Klein, but for configuration space it is equivalent to $\bb G\circ \hat H$. 

Taking these possibilities into account,  we possess all the ingredients to turn the crank of Section~\S\ref{sec:KK_st}. In particular, we can compute the equations of motion for the full configuration space, relating them to the Ehresmann bundle curvature and conserved quantities. To do this, we compute the equations of motion for the base space Lagrangian: 
$$\mathcal{L}_H:=K_H-V\,,$$
and then we can compare, as we did in \eqref{eq:compare}, the equations of motion leading to the acceleration:
\be\label{eq:compare_config}\pi_*(\frac{\D^{\F_T}\dot\gamma}{dt}) \qquad\text{to}\qquad \frac{\D^{\mathcal{S}}\dot\gamma}{dt}
\ee 
in both the translationally reduced configuration space, $\F_T\simeq \R^{3N-3}\simeq \mathcal{Q}/\R^3$, and the \emph{relative configuration space}, $\mathcal{S}\simeq \mathcal{Q}_T/\SO(3)$.

We arrived at \eqref{eq:geodesic_KK} using precisely the same considerations. The only difference here is the meaning of the vertical and the horizontal directions. Note that in neither case we are working within a section, or a choice of frames, etc. Instead, we are working with a choice of an anholonomic basis for the tangent bundle to configuration space, which we are splitting into horizontal and vertical directions.

Our notation will differ from that of space-time Kaluza--Klein theory: we use capital Roman letters,  $I, J$, spanning $3N-6$ horizontal directions  (which are to be projected to base space), and we will maintain early Roman letters; $a, b$; for the vertical directions spanning three vertical directions. We thus replace the $e_\sigma$ used in space-time Kaluza--Klein by   
$$e_a:=J_{(a)}^\#=\boldsymbol {n}_{(a)}\times \cc_\alpha$$
 (to be understood as the $n$-tuple). We need not specify a basis for the horizontal vectors, which, according to \eqref{eq:vert_proj}, are of the form: 
 $\dot\cc_\alpha- \sum_\alpha m_\alpha \dot\cc_\alpha\times \cc_\alpha\,.$ Instead, we will label them $e_I$. 

The fundamental vertical vectors are not normalized. For a geodesic, instead of \eqref{eq:geod_dec} we will therefore obtain:
\be \dot\gamma=u^I e_I +q^a e_a\,. 
\ee
Here, 
$$q^a=(M^{-1})^{ab}\bb G(e_a, \gamma')=(M^{-1})^{ab}L_b\,,$$ 
where $L_a$ is the angular momentum component of the motion. According to Theorem~\ref{ftnt:Killing}, it is only $\bb G(e_a, \gamma')=:L_a$ that is conserved along the geodesic; i.e., is such that $\nabla_{\gamma'}\bb G(e_a, \gamma')=0$ for $\gamma$ a geodesic. Therefore, in the geodesic equation, this difference  will create another purely vertical term of the form $(\nabla_{\gamma'}q^a)e_a$. This gets projected out when we contract with $e_I$. 

There are additional differences in the computation of the geodesic motion because the Christoffel symbols will differ from those of  \eqref{eq:main}.  More precisely, we obtain the Christoffel symbols analogous to \eqref{eq:main}, but the base-space dependence of the vertical metric gives us new ones as well. The computations are laborious but straightforward (see e.g. \citep[eq. 5.59]{LittlejohnReinsch}). We will explicitly need only two of them: 
\be
\begin{cases}\w^I_{ab}=-\frac12M_{ab}{}^{; I}\\
\w^I_{Ja}=\w^I_{aJ}=-\frac12 M_{ab}F^{bI}_J\,,
\end{cases}
\ee 
where we used a semi-colon to denote the covariant derivative with respect to $\varpi$; i.e., $M_{ab}{}^{; I}=\bb G^{-1}(\dd_\varpi M, \cdot)$. $\dd_\varpi$ is the gauge-covariant exterior derivative; i.e., $\dd_\varpi M=\dd M-[M, \varpi]$, and  $M$ is the moment of inertia tensor given in \eqref{eq:MIT}.

Putting it all together, we obtain the analogue of  \eqref{eq:Lor_dev}:
\begin{eqnarray}
e^I({\nabla^P}_{\dot\gamma}\dot\gamma)&=&\dot{u}^I+u^J (u^K \w ^I_{KJ}+q^a\w^I_{aJ})+q^b (u^K\w^I_{K b}+ q^a\w^I_{ab}) \nonumber\\
&=& (\dot{u}^I+u^J u^K \bar \Gamma ^I_{KJ})- u^J q^ aM_{ab} F^{bI}_J -\frac12 q^aq^b M_{ab}{}^{;I}\nonumber\\
&=& (\dot{u}^I+u^J u^K \bar \Gamma ^I_{KJ})- u^J L_b F^{bI}_J -\frac12 L^a L^b (M^{-1})_{ab}^{;I}\,,
\end{eqnarray}
where the  term inside the parenthesis in the last two lines is just the intrinsic geodesic of the base space. Adding a gauge-invariant potential term $V$ --- e.g. a potential term that depends only on inter-particle separations --- would only add a gradient $V_{; I}$ to the equations of motion.

 In abridged notation, and in  analogy to \eqref{eq:KK_Lorentz}, we obtain: 
$$
\bb G(\pi_*(\frac{\D^{\F_T}\dot\gamma}{dt})- \frac{\D^{\mathcal{S}}\dot\gamma}{dt}, \cdot)= \mathbf{L}\cdot  \bb F(\dot \cc, \cdot)-\frac12  \mathbf{L}\cdot \dd_\varpi M^{-1}\cdot  \mathbf{L}\,,
$$
where the square brackets are matrix commutators, $\bb F$ is the bundle curvature, given in \eqref{eq:curv_rot}, and $\mathbf{L}$ is the angular momentum. For convenience, the equation was been written as a 1-form in configuration space (i.e. with $\langle v|=\bb G(|v\rangle, \cdot)$). Since $\bb F$ is horizontal, we can rewrite this difference more clearly in terms of the deviation $\Delta_I$ using $\cdot$ to indicate the Euclidean inner product in $\bb R^3$:
\be\label{config_Lorentz}
\Delta_I=\mathbf{L}\cdot  F_{IJ}\dot \gamma^J-\frac12  \mathbf{L}\cdot \nabla_IM^{-1}\cdot  \mathbf{L}\,,
\ee
which is \emph{our main result}.

The first term on the right hand side of \eqref{config_Lorentz} arises from the curvature of the bundle, and is clearly identifiable as a `Lorentz-like' force. The second term arises from the base space dependence of the vertical kinetic term, and is a generalized centrifugal force\footnote{We call this term the centrifugal force because, in the two-body case, it is just $\frac{\d}{\d r}(\frac{ |\mathbf{L}|^2}{2mr^2})$, which is the centrifugal force for the radial equation of motion as derived in Section~\S\ref{sub:the_two_body_system}. \label{ftnt:Coul} } coming from the mass-like potential, $\mathbf{L}\cdot M^{-1}\cdot  \mathbf{L}$, quadratic in $\mathbf{L}$.\footnote{Equation \eqref{config_Lorentz} is only concerned with the difference between geodesic motion intrinsic to relative configuration space and the projection of geodesic motion on the bundle. The intrinsic geodesic equation is the projected one for zero angular momentum, i.e., (cf. \citet[Eq. 4.77]{LittlejohnReinsch}):
$$ \frac{\D^{\mathcal{S}}\dot\gamma_I}{dt}=\ddot \gamma_I+\Gamma^J_{IK}\dot\gamma^K\dot \gamma_J-\nabla_I V\,,
$$
where $\Gamma^I_{JK}$ are the Christoffel symbols of the horizontally projected metric  \citet[5.59]{LittlejohnReinsch}, and, again, $I, J, K$ parametrize coordinates $q^I$ in relative configuration space, $\mathcal{S}$. These are the explicit equations of motion on the reduced space with zero angular momentum as advocated in \cite{Barbour_Bertotti}. }

The Lorentz-like force term of \eqref{config_Lorentz} seems at first somewhat surprising: what is usually labelled a fictitious force --- the Coriolis effects is dependent on the frame chosen --- makes an appearance on the base space. But this is precisely what  ``falling cat experiments'' --- showing a possible change of orientation even for a motion with vanishing angular momentum --- implies. For bodies with non-zero angular momentum, such a force will appear in any frame. The kinetic connection-form in the rotational bundle has curvature, and when we couple that to a motion with non-vanishing angular momentum, we obtain a generalized Coriolis force.

Following the method in the proof to Theorem~\ref{ftnt:Killing} --- which shows that the vertical velocity is conserved by geodesic motion --- it is easy to show that vertical velocity must also be conserved here. As mentioned, the second term on the right hand side of \eqref{config_Lorentz} goes beyond an analogue of the original Lorentz-force deviation in space-time-based Kaluza--Klein; it comes from the configuration space dependence of the kinematical metric along the orbits. For comparison, in the space-time case, the vertical inner product is induced by the Killing form on the Lie-algebra  (cf. \eqref{kil}), as it is for the translations ($\bb Q_{ab}=\delta_{ab}$), and thus in both these cases the vertical energy is independent of the base space point. For the rotations, the moment of inertia tensor $\bb Q_{ab}=M_{ab}$,  which is not independent of the relative configurations. The vertical kinetic term, $\mathbf{L}\cdot M^{-1}\cdot  \mathbf{L}$ then plays a role of a potential on relative configuration space.

In sum: because the rotational invariance of the Lagrangian is only broken by time-dependent transformations, we are able to find a suitable bundle structure under which to study the projection of its dynamics. But the dynamics can only be fully represented in the quotient space with the introduction of more structure: the curvature-form and the angular momentum. Nonetheless, as we saw in \eqref{eq:curv_rot}, the only quantity required for the representation of the projected dynamics that is not intrinsic to relative configuration space is the angular momentum, which is easily shown to be a constant of motion.  For translations, since $\bb F=0$ and $\bb Q_{ab}=\delta_{ab}$, the analogue to \eqref{config_Lorentz} vanishes. 

On base space points for which the moment of inertia, $M_{ab}$ is the identity, the centrifugal potential becomes a Casimir invariant, which does not affect the equations of motion. For zero angular momentum, the difference between the projected dynamics and the intrinsic dynamics also vanishes: there are no charges on which the effective potential and the Lorentz force act. 
 
More generally, we conjecture that two terms of \eqref{config_Lorentz} are the most general form of terms that can arise from the reduction of a theory whose Lagrangian is $\G$-invariant separately in its potential and kinetic energy terms when the connection-form is defined by orthogonality. Specifically, we conjecture that the possible terms that can be obtained through such a reduction include one extra term coming from the curvature of the bundle and another coming from the base-space variance of the vertical inner product.

\section{Conclusions}	\label{sec:conclusions}

\subsection{Technical summary}

Let us take a moment to summarise the key technical achievements of the paper. At the representational level, if a time-independent Lagrangian symmetry is specified, it is possible to produce a time-dependent generalisation of the symmetry by introducing a connection-form following the procedure outlined in Section~\S\ref{sec:KK_config}. Such a connection-form is singled out in that it allows an efficient organization and interpretation of the reduced equations of motion.

When the configuration space is endowed with a group-invariant metric, imposing orthogonality with respect to the orbits of the group action provides a canonical notion of horizontality --- or parallel transport --- between representatives of the system in the different orbits. Moreover, the identification of the kinetic term in the Lagrangian with the norm of the velocities under this inner-product, when conjoined to a fully-invariant potential term, ensures: (1) that the horizontal contribution to the dynamics is fully gauge-invariant (following the arguments at the beginning of Section~\S\ref{sec:KK_proj}) and (2) that the vertical velocities are dynamically conserved (following Theorem~\ref{ftnt:Killing} of Section~\S\ref{sub:gen Lorentz} and its caveats). Using the connection-form constructed in this way, it follows that  the projection of the full configuration-space dynamics onto the base space requires only one additional constant of motion for each conserved charge generated by the symmetry group.

The natural physical interpretation of the geometrical quantities introduced here is illuminating. Conserved charges are naturally represented by vertical velocities in the bundle. Zero total angular and linear momentum are represented by horizontal motion along rotational or translational fibres respectively. The Lorentz-like force is represented by the first term of \eqref{config_Lorentz}, which arises in this picture from the bundle curvature \eqref{eq:curv_formula}.  Centrifugal force is represented by the extra mass-like term (quadratic in $\mathbf{L}$) of \eqref{config_Lorentz}, which arises from the base-space dependence of the vertical metric \eqref{eq:MIT}.

Notably, both these novel terms act on the angular momentum charge and vanish when $\mathbf{L}=0$. In the case of translations, the analogous charge (i.e., the total linear momentum) can always be set to zero without consequence. In the standard Kaluza--Klein treatment for semi-simple Lie groups over space-time, the fibre inner product is rigid, and the second potential term does not appear. These two terms therefore capture all the novel representational features of rotation.

\subsection{Conclusions and prospectus}

Our analysis has identified the formal differences between translation and rotation that have made rotation such a difficult problem for the relationalist. We can understand these differences by considering that under translations there is no difference between the intrinsic dynamics in the reduced space and the projected dynamics from the configuration space. One obtains conserved charges but there is no curvature or extra potential to act upon them. Motion on the reduced space does not carry any information about the (constant) vertical velocity in the full configuration space and is therefore easily expressed in relational terms.

In contrast, under rotations one acquires two non-trivial terms: a Lorentz-force-like term due to curvature and another mass-like quadratic potential term mediated by the moment of inertia (see footnote \ref{ftnt:Coul}). It should be emphasised that both these terms have a simple dependence on the natural geometric structures defined entirely on the relative configuration space and the straightforward coupling of these structures to an $SO(3)$ charge. What the simple and elegant form of these terms illustrates is that an account of Coriolis and centrifugal effects due to rotation in absolute space, while always available, is certainly not necessary and may not even be the most appealing account available.

The construction we have given in terms of principal fibres bundles can help to clarify some of the limitations of the symplectic reduction of \cite{marsden1992lectures} and the frame-based approach of \cite{LittlejohnReinsch}. These approaches make use of frames defined in terms of the \emph{instantaneous} configurations of the system. But because the rotational bundle is curved, no frame of this kind can be found that will conserve angular momentum when it is non-zero.  In the fibre bundle picture however one can always employ an anholonomic frame along the geodesics of the kinematical metric, so that angular momentum is guaranteed to be conserved by Theorem~\ref{ftnt:Killing}. This gives a \emph{dynamically} defined inertial frame in terms of relational quantities.

The analysis of Section~\S\ref{sec:LR} not only gives us the tools required to identify the novel features that distinguish rotation from translation or standard Kaluza--Klein theories; it also gives us the tools required to write down and assess a theory of rotation in relational terms. Three important considerations arise immediately from the current work.

First, the formal analogies between translation and rotation symmetry in Newtonian mechanics and $U(1)$ symmetry in electromagnetism suggest two different perspectives on conserved charges: that of the reduced space where the conserved charge --- i.e., linear momentum, angular momentum, and electric charge respectively --- is additional structure that must be specified; and that of the relevant bundle where they are completely geometrized. Using these analogies, it is possible to fruitfully compare the relational versus absolute positions of Newtonian mechanics to their electromagnetic counterparts in a way that was not possible before the construction presented here. 

Second, in Kaluza--Klein theory, the motion of a charged particle in an electromagnetic field can be geometrically  described without invoking either the charge or the electromagnetic force. One says that charged particle motion arises merely from geometrical properties of the (higher-dimensional) space in which the particle lives. Be that as it may be, most physicists today reject the geometrized origin in favour of postulating one more field --- the electromagnetic field ---  alongside its charge.  While a $U(1)$ fibre bundle is sometimes used to model electromagnetic effects, the fibres are interpreted as `internal' directions of a very different character than the extra space-time dimenions of Kaluza--Klein theory. The situation seems entirely analogous for rotations and angular momentum: one can have a description that geometrizes the motion in absolute space or one can use an enriched relational ontology that postulates new fields and corresponding conserved charges with remarkably simple interaction terms. 
 
Third, the appearance of charge-dependent terms in the base-space dynamics suggests that the most appropriate ontology for describing rotations, and indeed any phenomena that can be modelled with the Kaluza--Klein-inspired construction presented here, is an \emph{enriched ontology} that adds extra conserved charges to a purely relational theory. It is precisely the conservation of these charges that permits the use of such an enriched structure. If the phenomena \emph{could not} be modelled using conserved charges but \emph{could} be modelled using time-varying charges, then the entire projection procedure would be invalidated. Clearly, the justification for this projection is different in standard Kaluza--Klein theory --- where it was the epistemic (``too small to see'') argument we rehearsed in the Introduction --- and in Newtonian mechanics. But if the phenomena could not be saved with conserved charges alone, or if the equations of motion on the reduced space were uninformative and practically useless, there would be a compelling case against the relationalist. 

Similarly, the absolutist must justify why the relevant charge is observed to be conserved in time or, equivalently, why a rotational symmetry should be imposed at the level of the action. The absolutist must say why the world appears to accommodate significantly less possibilities than those that are allowed by their ontology. A full analysis of these compelling considerations will have to wait for the more complete treatment we will give in our second paper.

 \section*{Acknowledgements}

We would like to extend heartfelt appreciation to Jeremy Butterfield for extensive meticulous comments on the paper. His guidance was indispensable to the final outcome of the paper. We would also like to thank Bryan Roberts and Karim Th\'ebault for helpful suggestions and encouragement. HG would like to especially thank Robert Littlejohn for clarifying many questions regarding the work of \cite{LittlejohnReinsch}. We would like to thank Erik Curiel for his inventive title suggestion: ``Charging the Relationalist for Angular Momentum,'' which we used as the section title for Section~\S\ref{sec:LR}. Finally, we would like to thank Julian Barbour for bringing us together many years ago in College Farm to discuss relational issues in physics.

\begin{appendix}
\section{Auxiliary computations for the connection-form of a bundle $P$}\label{app:connection}

A connection in  $P$ is a sub-bundle $\HH$ of  $TP$ such that $g_*(\HH_p)=\HH_{g\cdot{p}}$ and $\HH_p\oplus\V_p=T_pP$. We call $\HH$ the horizontal bundle.

First, since the vertical space is defined by the span of $\iota$, as defined in \eqref{eq:iota}:
\begin{eqnarray}
\iota_p\left(\Ad(g^{-1})\xi\right)=\frac{d}{dt}{}_{|t=0}({\exp}(t\left(\Ad(g^{-1})\xi\right))\cdot{p})=\frac{d}{dt}{}_{|t=0}
(g^{-1}{\exp}(t\xi)g\cdot{p})\nonumber\\ 
~~\therefore~~g_*(\iota_p\left(\Ad(g^{-1})\xi\right))=\frac{d}{dt}{}_{|t=0}({\exp}(t\xi)g\cdot{p})=\iota_{g\cdot{p}}(\xi)\label{eq:iota_equi}\end{eqnarray}  
Thus $g_*(\V_p)=\V_{g\cdot{p}}$. Therefore the decomposition $T_pP=\HH_p\oplus\V_p$ is $G$-invariant. 



As in the main text, we denote the smooth projections from  $TP$ to the  horizontal bundle as $\widehat{H}$ and mutatis mutandis for the vertical bundle $\widehat{{V}}$. Clearly, from $G$-invariance of the decomposition,   $\widehat{H}_{g\cdot{p}}\circ g_*=g_*\circ\widehat{H}_p$, and mutatis mutandis for the vertical projection. Explicitly, if $w\in{T_pP}$ then, since $w=w_h+w_v=\hat{H}(w)+\hat{V}(w)$ and, since $G_*$ is linear, and $g_*(w_v)\in\V_{g\cdot{p}}$, we automatically obtain $g_*(w_h)\in\HH_{g\cdot{p}}$.  In other words,  $\widehat{H}_{g\cdot{p}}g_*(w)=g_*(\hat{H}(w))$. 
We find: \begin{theo}
If $\HH$ is a connection on $P$, then for all $p\in{P}$: 
$\pi_p{}_*:T_pP\ra{T_{\pi(p)}M}$ restricts to a linear isomorphism $h_p:\HH_p\ra{T_{\pi(p)}M}$ such that $h_{g\cdot{p}}\circ g_*{}{|_{\HH_p}}=h_p$.\end{theo}
\begin{proof}
In finite dimensions, it is easy to show that the projection restricts to a linear isomorphism (for infinite-dimensions cf. \citep{Ebin, YangMillsSlice}), since the vertical space is in the kernel of the projection, which is of maximal rank, i.e. surjective. Moreover, $h_{g\cdot{p}}\circ\widehat{H}_{g\cdot{p}}=\pi_*{}_{g\cdot{p}}$, therefore: 
\begin{gather*}
h_{g\cdot{p}}\circ\widehat{H}_{g\cdot{p}}\circ\pi_*{}
=\pi_*{}_{g\cdot{p}}\circ\pi_*{}=\pi_*{}_p=h_p\circ\widehat{H}_p\\
\therefore~h_p\circ\widehat{H}_p=h_{g\cdot{p}}\circ\pi_*{}\circ\widehat{H}_{p} 
\end{gather*}
and since $\widehat{H}_p$ is surjective, we obtain the theorem.
\end{proof}

For a given $\HH$, we define the connection-form as 
$\omega:=\iota^{-1}\circ\widehat{V}$. That is:
\begin{align*}
\omega_p:~& T_pP\ra\fg\\
\phantom{\omega_p:}~& ~u\mapsto\iota_p^{-1}\circ\widehat{V}_p(u) 
\end{align*}

Of course, for $v\in\V_p$ then $\omega(v)=\iota^{-1}(v)$.  Moreover: 
\begin{gather*}
(g^*\omega)_p=\omega_{g\cdot{p}}\circ\pi_*{}=\iota^{-1}_{g\cdot{p}}\circ\widehat{V}_{g\cdot{p}}\circ\pi_*{}\\ 
=\iota^{-1}_{g\cdot{p}}\circ\pi_*{}\circ\widehat{V}_p=(g^*\iota^{-1})_p\circ{\widehat{V}}_p=\Ad(g)\circ\iota^{-1}_p\circ\widehat{V}_p=\Ad(g)\omega_p\\
\end{gather*} and therefore
\be\label{foconex}~~~~g^*\omega=\Ad(g)\omega\ee 
That shows  $\omega$ is a connection-form.

Conversely, given a connection-form, we can find the corresponding (equivariant) horizontal space through its kernel.  That is, since $\iota_p$ is a linear isomorphism over $\V_p$~, clearly $ \Ker(\omega_p)$ complements $\V_p$~. Moreover, if $u\in\Ker(\omega_p)$ then  $g_*(u)\in\Ker(\omega_{g\cdot{p}})$ since 
$$ \omega_{g\cdot{p}}\circ\pi_*{}(u)=\Ad(g)\omega_p(u)=0.$$
Therefore, since $g_*$ is a linear isomorphism, $\Ker(\omega)$ is a  $G$-invariant sub-bundle of $TP$, complementing $\V$. That is, we can define the horizontal space as $\HH_p=\Ker\omega_p$. Indeed, there is a bijective correspondence: $\HH\leftrightarrow{\omega}$. 

\end{appendix}

\bibliographystyle{apacite} 
\bibliography{references}
\end{document}